\documentclass[runningheads]{llncs}

\usepackage{amssymb}
\usepackage{stackrel}
\usepackage{amsmath}
\usepackage{tikz}
\usepackage{tkz-euclide}
\usepackage{ifthen}
\usetkzobj{all}
\usetikzlibrary{automata,positioning,arrows,shapes,math,arrows.meta}
\tikzset{
	big arrow/.style={
		decoration={markings,mark=at position 1 with {\arrow[scale=3,#1]{>}}},
		postaction={decorate},
		shorten >=0.4pt}}

\newcommand{\by}{\ensuremath{\bold{y}}}
\newcommand{\bz}{\ensuremath{\bold{z}}}
\newcommand{\bu}{\ensuremath{\bold{u}}}
\newcommand{\bv}{\ensuremath{\bold{v}}}
\newcommand{\bw}{\ensuremath{\bold{w}}}
\newcommand{\bm}{\ensuremath{\bold{m}}}
\newcommand{\bn}{\ensuremath{\bold{n}}}

\newcommand{\bi}{\ensuremath{\bold{i}}}
\newcommand{\bj}{\ensuremath{\bold{j}}}

\newcommand{\zv}{\ensuremath{\bold{0}}}

\newcommand{\A}{\ensuremath{\mathcal{A}}}
\newcommand{\B}{\ensuremath{\mathcal{B}}}

\newcommand{\calP}{\ensuremath{\mathcal{P}}}
\newcommand{\calR}{\ensuremath{\mathcal{R}}}

\newcommand{\M}{\ensuremath{\mathcal{M}}}
\newcommand{\N}{\ensuremath{\mathbb{N}}}
\newcommand{\calN}{\ensuremath{\mathcal{N}}}
\newcommand{\Z}{\ensuremath{\mathbb{Z}}}
\newcommand{\R}{\ensuremath{\mathbb{R}}}
\newcommand{\Q}{\ensuremath{\mathbb{Q}}}
\newcommand{\Exp}{\ensuremath{\mathbb{E}}}
\newcommand{\bigO}{\ensuremath{\mathcal{O}}}
\newcommand{\ce}[1]{\ensuremath{\left(#1 \right)}}

\newcommand{\size}[1]{|\!|#1|\!|}

\newcommand{\term}{\mathcal{L}}
\newcommand{\conf}{\mathit{C}}

\newcommand{\Dec}{\mathit{Dec}}
\newcommand{\Sl}{\mathit{SimLen}}
\newcommand{\Ipath}{\mathit{Ipath}}
\newcommand{\Spath}{\mathit{Spath}}

\newcommand{\Term}{\mathit{Term}}

\newcommand{\minup}{\min_{\mathcal{A}}}
\newcommand{\maxup}{\max_{\mathcal{A}}}
\newcommand{\eps}{\varepsilon}
\newcommand{\tran}[1]{\stackrel{#1}{\rightarrow}}
\newcommand{\MP}{\ensuremath{\mathrm{MP}}}

\newcommand{\mec}{M}

\begin{document}
\title{Deciding Fast Termination for Probabilistic VASS with Nondeterminism\thanks{Tom\'{a}\v{s} Br\'{a}zdil and Anton\'{\i}n Ku\v{c}era are supported by the Czech Science Foundation Grant No.~18-11193S. Krishnendu Chatterjee is supported by the Austrian Science Fund (FWF) NFN Grants S11407-N23 (RiSE/SHiNE). Petr Novotn\'y and Dominik Velan are supported by the Czech Science Foundation Grant No.~GJ19-15134Y.}} 

\titlerunning{Deciding Fast Termination for Probabilistic VASS}

\author{Tom\'{a}\v{s} Br\'{a}zdil\inst{1} \and
Krishnendu Chatterjee\inst{2} \and
Anton\'{\i}n Ku\v{c}era\inst{1}
\and
Petr Novotn\'{y}\inst{1} \and
Dominik Velan\inst{1}}

\authorrunning{T.~Br\'{a}zdil et al.}

\institute{Faculty of Informatics, Masaryk University\\
\email{\{xbrazdil,tony,petr.novotny,xvelan1\}@fi.muni.cz}\and
IST Austria\\
\email{krishnendu.chatterjee@ist.ac.at}} 

\sloppy
\maketitle

\begin{abstract}
A probabilistic vector addition system with states (pVASS) is a finite state Markov process augmented with non-negative integer counters that can be incremented or decremented during each state transition, blocking any behaviour that would cause a counter to decrease below zero. The pVASS can be used as abstractions of probabilistic programs with many decidable properties. The use of pVASS as abstractions requires the presence of nondeterminism in the model. In this paper, we develop techniques for checking fast termination of pVASS with nondeterminism.
That is, for every initial configuration of size n, we consider the worst expected number of transitions needed to reach a configuration with some counter negative (the expected termination time). We show that the problem whether the asymptotic expected termination time is linear is decidable in polynomial time for a certain natural class of pVASS with nondeterminism. Furthermore, we show the following dichotomy: if the asymptotic expected termination time is not linear, then it is at least quadratic, i.e., in $\Omega(n^2)$. 

\end{abstract}

\section{Introduction}
\label{sec-intro}

\textbf{Probabilistic Programs \& VASS} Probabilistic systems play an important role in various areas of computing such as machine learning \cite{Ghahramani:15-prob-ai-nature}, network protocol design~\cite{FosterKMRS:16-prob-netkat}, robotics~\cite{ThrunBurgardFox:2005-robotics}, privacy and security~\cite{BartheGGHS:16-diff-privacy-prob-couplings}, and many others. For this reason, verification of probabilistic systems receives a considerable attention of the verification community. As in the classical (non-probabilistic) setting, in probabilistic verification one typically constructs a suitable abstract model over-approximating the real behaviour of the system. In the past, the verification research was focused mostly on finite-state probabilistic models~\cite{BK:book} as well as some special infinite-state classes, such as probabilistic one-counter~\cite{BKK:pOC-time-LTL-martingale-JACM} or pushdown automata~\cite{EKM:prob-PDA-PCTL-LMCS,EY:RMC-LTL-complexity-TCL}. However, the recent proliferation of general, Turing-complete \emph{probabilistic programming languages} (PPLs) necessitates the use of more complex models, that can encompass multiple potentially unbounded numerical variables.

In the classical setting, one of the standard formalisms used for program abstraction are \emph{vector addition systems with states (VASS)}~\cite{KM69}. Intuitively, a VASS is a finite directed graph where every edge is assigned a vector of integer counter updates of a fixed dimension~$d$. A \emph{configuration} $p\bv$ is specified by a current state $p$ and a vector of current counter values $\bv$. The computation proceeds by moving along the edges in the graph and performing the respective updates on the counters. Since VASS themselves are not Turing-complete, they have many decidable properties, and they have been successfully used as program abstractions in termination and complexity analysis~\cite{SZV14} as well as for reasoning about parallel programs~\cite{EN94,KM69} and parameterized systems~\cite{Bloem16,conf/icalp/AminofRZS15}. Applying such an abstraction to a probabilistic program yields a \emph{probabilistic VASS (pVASS)}, which allows for a \emph{probabilistic choice} of a transition in some states. Moreover, during the abstraction, certain complex programming constructs such as \texttt{if-then-else} branching are replaced with \emph{nondeterministic choice}. To ensure that the abstraction over-approximates the possible behaviour, we typically interpret the nondeterminism as \emph{demonic}, i.e., the choice is resolved by adversarial environment. However, in certain settings it makes sense to consider \emph{angelic nondeterminism}, to be resolved by a yet-to-be-designed controller (e.g., a scheduling mechanism in a queuing system).

\noindent
\textbf{Termination Complexity} One of the fundamental problems in program analysis is to evaluate a given program's runtime. In the classical setting, this problem emerges in various flavours, ranging from worst-case execution time-analysis~\cite{Wilhelm&al:2008:wcet-survey,Cassez:11-wcet-survey} in real-time systems to obtaining bounds on the number of execution steps~\cite{GulwaniMehraChulimbi:09-speed}, analysing asymptotic~\cite{CFG17}, or amortized complexity~\cite{journals/toplas/0002AH12}. VASS-based abstractions were successfully used in the latter scenario~\cite{SZV14}. 

Recently, several approaches to reason about the expected runtime of probabilistic programs were developed~\cite{KaminskiKMO:18-wp-expected-runtime-jacm,NgoCarbonneauxHoffmann:18-pp-resource-analysis}. The analysis is much more demanding than in the classical case. For instance, deciding whether the expected runtime is finite is harder (i.e. higher in the arithmetic hierarchy) than deciding whether a probabilistic program terminates with probability one~\cite{KaminskiKatoenMateja:19-pp-hardness-act-inf}. Additional obstacle is the inherent \emph{non-compositionality} of expected runtimes. The work~\cite{KaminskiKMO:18-wp-expected-runtime-jacm} gives an example of two programs, $P_1$, and $P_2$, which both consist of a single loop (i.e. they have a strongly connected control flow graph) and whose expected runtime is \emph{linear} in the magnitude of initial variable valuations; but running $P_2$ after $P_1$ yields the program $P_1; P_2$ whose expected runtime is \emph{infinite}.

These intricacies spawn fundamental questions about probabilistic models, which we aim to address: \emph{Is there a sufficiently powerful probabilistic formalism where a fast (i.e., linear-time) termination from an arbitrary initial configuration is decidable? Can the decision procedure proceed by analysing individual strongly-connected components and composing the results? Can we provide a lower bound on the expected runtime in the case that it is not linear?} These questions were previously considered in the non-probabilistic setting, namely in the domain of VASS~\cite{BCKNVZ:VASS-termination-LICS}. In this paper, we investigate them in the probabilistic context.

\noindent
\textbf{Our Setting} 
We show that the above questions can be answered affirmatively in the domain in pVASS with nondeterminism, which are Markov decision processes over VASS where the nondeterministic choice is resolved either demonically (i.e. the nondeterminism tries to prolong the computation) or angelically. We consider a basic variant of VASS termination: the \emph{zero termination}, where the computation stops when some counter becomes negative. The \emph{termination complexity} of a given pVASS is a function $\term \colon \N \rightarrow \N \cup \{\infty\}$ assigning to every $n$ the maximal/minimal (in the demonic/angelic case) expected length of a computation initiated in a configuration of size $n$ (the size of $p\bv$ is defined as the maximal component of $\bv$), where the maximum/minimum is taken over all the strategies of the environment (we consider unrestricted, i.e., history-dependent and randomized, strategies).


%

\noindent
\textbf{Our Results} For \emph{strongly connected} pVASS which contain either a demonic or an angelic non-determinism (but not both) we show that
\begin{enumerate}
	\item The problem whether $\term \in \bigO(n)$ is decidable in \emph{polynomial time}.
	\item If $\term \not\in \bigO(n)$, then $\term \in \Omega(n^{2})$.
	\item If $\term \not\in \bigO(n)$, then for every $\varepsilon > 0$, the probability of all computations of length at least $n^{2-\varepsilon}$ converges to one as $n \rightarrow \infty$, (in the demonic case, this requires  the environment to use appropriate strategies). 
\end{enumerate}
According to \textbf{2.},  $\term \not\in \bigO(n)$ implies that $\term$ is ``at least quadratic''. However, \textbf{3.}\  does not follow from~\textbf{2.} (a more detailed discussion is postponed to Section~\ref{sec-results}). 

We also show that the above results hold in general VASS with angelic nondeterminism, while in the demonic setting they extend to a restricted class pVASS whose maximal end-component (MEC) decomposition yields a directed acyclic graph (DAG), in which case  \textbf{1.} can be solved compositionally by analysing individual MECs. Finally, we show that in pVASS whose MEC-decomposition is not DAG-like, the demonic complexity \emph{cannot} be decided by analysis of individual MECs, since such VASS can emulate the non-compositional example of~\cite{KaminskiKMO:18-wp-expected-runtime-jacm}.

The results build on analogous results for non-probabilistic VASS established in \cite{BCKNVZ:VASS-termination-LICS}, combining them with a novel probabilistic analysis.
\smallskip 

\noindent
\textbf{Paper Organization.} After presenting preliminaries in Section~\ref{sec-prelim}, we focus on the demonic case which contains the main technical contributions. Subsection~\ref{subsec:outline} provides an intuitive outline of our techniques. Subsection~\ref{sec-linear} develops the algorithm for proving linear termination complexity and shows its soundness (i.e. that a yes-answer indeed proves $\term_d(n)\in\mathcal{O}(n)$). Subsection~\ref{sec-scheme} deals with the quadratic lower bound, showing the completeness of our algorithm, and Subsection~\ref{sec-angel} discusses extension of the results to the angelic case. Finally, in Section~\ref{sec:general} we extend the techniques to DAG-like VASS MDPs and discuss the difficulties arising in general VASS. 
Missing proofs are provided in the appendix.

\smallskip 

\noindent
\textbf{Related Work.} 
The termination problems (counter-termination, control-state termination) for classical VASS as well as the related problems of boundedness and coverability have been studied very intensively in the last decades, see, e.g., \cite{Lipton:PN-Reachability,Rackoff:Covering-TCS,Esparza:PN,ELMMN14:SMT-coverability,BG11:Vass}.
The complexity of the termination problem with fixed initial configuration is
EXPSPACE complete~\cite{Lipton:PN-Reachability,Yen92:Petri-Net-logic,AH11:Yen}.
The more general reachability problem is also decidable~\cite{Mayr:PN-reachability,Leroux:VASS-reachability-short,Kosaraju82:VASS-reach-dec}, but computationally hard 
\cite{Lipton:PN-Reachability,DBLP:conf/stoc/CzerwinskiLLLM19}. The best known upper bound is
Ackermannian~\cite{DBLP:journals/corr/abs-1903-08575} (see~\cite{Schmitz16:hyperackermannian-complexity-hierarchy}  for an overview of hyper-Ackermannian complexity hierarchies). 

The problem of existence of infinite computations in VASS has been also studied in
the literature. Polynomial-time algorithms have been presented in~\cite{CDHR10,VCDHRR15} using results of~\cite{KS88}. In the more general context of games played on VASS, even deciding the existence of
infinite computation is coNP-complete~\cite{CDHR10,VCDHRR15}, and various algorithmic approaches
based on hyperplane-separation technique have been studied in~\cite{CV13,JLS15,CJLS17}.

The study on asymptotic termination complexity of non-probabilistic VASS, initiated in~\cite{BCKNVZ:VASS-termination-LICS} was continued in \cite{Leroux:VASS-polynomial-term-comp}, where the existence of \emph{some} $k$ such that $\term \in \bigO(n^k)$ was also shown decidable in polynomial time.

Concerning expected runtime analysis, we note the work~\cite{ChatterjeeFuMurhekar:17-automated-recurrence} which presents a sound (but incomplete) technique for obtaining near-linear asymptotic bounds on recurrence relations arising from certain types of probabilistic programs. 

\vspace{-0.2cm}

\section{Preliminaries}
\label{sec-prelim}


We use $\N$, $\Z$, $\Q$, and $\R$ to denote the sets of non-negative integers,
integers, rational numbers, and real numbers. 
%
Given a function $f \colon \N \rightarrow \N$, we use  $\bigO(f(n))$ and $\Omega(f(n))$ to denote the sets of all $g \colon \N \rightarrow \N$ such that $g(n) \leq a \cdot f(n)$ and $g(n) \geq b \cdot f(n)$ for all sufficiently large $n \in \N$, where $a,b$ are some positive constants. If $h(n) \in \bigO(f(n))$ and $h(n) \in \Omega(f(n))$, we write $h(n) \in \Theta(f(n))$.

Let $A$ be a finite index set. The vectors of $\R^A$ are denoted by bold letters such as $\bu,\bv,\bz,\ldots$. The component of $\bv$ of index $i\in A$ is denoted by $\bv(i)$. 
If the index set is of the form $A=\{1,2,\dots,d\}$ for some positive integer $d$, we write $\R^d$ instead of $\R^A$. For every $n \in \N$, we use $\bn$ to denote the constant vector where all
components are equal to~$n$.
The scalar product of $\bv,\bu \in \R^d$ is denoted by $\bv \cdot \bu$, i.e.,
$\bv\cdot \bu = \sum_{i=1}^d \bv(i)\cdot\bu(i)$. The other
standard operations and relations on $\R$ such as
$+$, $\leq$, or $<$ are extended to $\R^d$ in the component-wise way. In particular,  $\bv < \bu$ if $\bv(i) < \bu(i)$ for every index $i$.
\vspace{-0.2cm}

\subsection{Markov Decision Processes}

\begin{definition}
	\label{def-VASS}
	Let $L$ be a set of \emph{labels}. A \emph{Markov decision process (MDP)} with $L$-labeled transitions is a tuple $\A = \ce{Q, (Q_n,Q_p),T,P}$, where $Q \neq \emptyset$ is a finite set of \emph{states} split into two disjoint subsets $Q_n$ and $Q_p$ of \emph{nondeterministic} and \emph{probabilistic} states, $T \subseteq Q \times L \times Q$ is a finite set of \emph{labeled transitions} such that  every $q \in Q$ has at least one outgoing transition, and $P$ is a function assigning to each $(p,\ell,q) \in T$ where $p \in Q_p$ a positive rational probability so that, for every $p \in Q_p$,  $\sum_{(p,\ell,q) \in T} P(p,\ell,q) =1$. 
\end{definition}

A state $q$ is an \emph{immediate successor} of a state $p$ if there is a transition $(p,\ell,q)$ for some $\ell \in L$. A \emph{finite path} in $\A$ of length~$n$ is a finite sequence of the form
$p_0,\ell_1,p_1,\ell_2,p_2,\ldots,\ell_n,p_n$ where $n \geq 0$ and
$(p_i,\ell_{i+1},p_{i+1}) \in T$ for all $0 \leq i < n$. If $n \geq 1$ and $p_0 = p_n$, then
$\pi$ is a \emph{cycle}. An MDP is \emph{strongly connected} if for each pair of distinct states $p,q$ there is a finite path from $p$ to $q$. An \emph{infinite path} in $\A$ is an infinite sequence $p_0,\ell_1,p_1,\ell_2,p_2,\ldots$ such that $p_0,\ell_1,p_1,\ldots,\ell_n,p_n$ is a finite path for every $n \geq 0$. For a finite sequence of the form $\pi = p_0,\ell_1,p_1,\ell_2,p_2,\ldots,\ell_n,p_n$ and a finite or infinite sequence of the form $\varrho = q_0,\kappa_1,q_1,\kappa_2,\ldots$, where $\pi$ and $\varrho$ are not necessarily paths in $\A$, we use $\pi \odot \varrho$ to denote the \emph{concatenated sequence} $p_0,\ell_1,\ldots,\ell_n,p_n,\kappa_1,q_1,\kappa_2,\ldots$ (we do not require $p_n = q_0$). If $\pi,\varrho$ are both paths in $\A$ and $p_n = q_0$, then $\pi \odot \varrho$ is also a path in $\A$.

A \emph{strategy} is a function $\sigma$ assigning to every finite path $p_0,\ell_1,p_1,\ldots,\ell_n,p_n$ ending in a nondeterministic state a probability distribution over the outgoing transitions of $p_n$. A strategy is \emph{Markovian (M)} if it depends only on the last state $p_n$, and \emph{deterministic (D)} if it always selects some successor state with probability one. The set of all strategies is denoted by $\Sigma$ (the underlying $\A$ is always clearly determined by the context). Every initial state $p \in Q$ and every strategy $\sigma$ determine the probability space over infinite paths initiated in $p$ in the standard way, and we use $\calP^\sigma_{p}$ to denote the associated probability measure.

\vspace{-0.1cm}

\subsection{Probabilistic VASS with Nondeterminism}

\begin{definition}
	Let $d \in \N$. A \emph{$d$-dimensional probabilistic VASS with non-determinism (VASS MDP)} is an MDP where the set of labels is $\Z^d$.
\end{definition}

Let $\A = \ce{Q, (Q_n,Q_p),T,P}$ be a $d$-dimensional VASS MDP. The encoding size of $\A$ is denoted by $\size{\A}$, where the integers representing counter updates are written in binary.  A \emph{configuration} of $\A$ is a pair $p\bv$, where $p \in Q$ and $\bv \in \Z^d$. If some component of $\bv$ is negative, then $p\bv$ is \emph{terminal}. The set of all configurations of $\A$ is denoted by $\conf(\A)$. The
\emph{size} of  $p\bv \in \conf(\A)$ is $\size{p\bv} = \size{\bv} = \max \{|\bv(i)| : 1\leq i \leq d\}$. Given $n\in \N$, we say that $p\bv$ is \emph{$n$-bounded} if $\size{p\bv}\leq n$. 

Every (finite or infinite) path $p_0,\bu_1,p_1,\bu_2,p_2,\ldots$ and every initial vector $\bv \in \Z^d$ determine the corresponding \emph{computation} of $\A$, i.e., the sequence of configurations $p_0\bv_0, p_1 \bv_1, p_2 \bv_2,\ldots$ such that $\bv_0 =\bv$ and $\bv_{i+1} = \bv_i + \bu_{i+1}$. For every \emph{infinite} computation  $\pi = p_0 \bv_0,p_1 \bv_1,p_2 \bv_2,\ldots$, let $\Term(\pi)$ be the least $j$ such that $p_j\bv_j$ is terminal. If there is no such $j$, we put  $\Term(\pi) = \infty$ .

Recall that every strategy $\sigma$ and every $p \in Q$ determine a probability space over infinite paths initiated in~$p$ with probability measure $\calP_p^\sigma$. Similarly, $\sigma$~determines the unique probability space over all \emph{computations} initiated in a given $p \bv$, and we use $\calP^\sigma_{p \bv}$ to denote the associated probability measure, and $\Exp^\sigma_{p\bv}[\Term]$ denotes the expected value of $\Term$.

The \emph{angelic/demonic termination complexity} of $\A$ are the functions $\term_a,\term_d \colon \N \rightarrow \R \cup \{\infty\}$ defined as follows, where $\conf_n(\A)$ is the set of all $p\bv \in \conf(\A)$ such that $\size{p\bv} = n$:
\vspace{-1mm}
\begin{eqnarray*}
\term_a(n) & \ = \ & \max_{p\bv \in \conf_n(\A)} \ \inf_{\sigma \in \Sigma}\ \Exp^\sigma_{p\bv}[\Term],\\
\term_d(n) & \ = \ & \max_{p\bv \in \conf_n(\A)} \ \sup_{\sigma \in \Sigma}\ \Exp^\sigma_{p\bv}[\Term].
\vspace{-1mm}
\end{eqnarray*}
We say that the expected angelic/demonic termination time of $\A$ is \emph{linear} if $\term_a(n) \in \bigO(n)$ and $\term_d(n) \in \bigO(n)$, respectively.

\section{Linearity of Demonic Termination Time}
\label{sec-results}

In this paper, we prove the following theorem:

\begin{theorem}
\label{thm-main}
  The problem whether the expected termination time of a given strongly connected VASS MDP $\A$ is linear is decidable in polynomial time. If the expected termination time of $\A$ is \emph{not} linear, then $\term_d(n) \in \Omega(n^2)$. Furthermore, for every $\varepsilon > 0$ we have that
\vspace{-1mm}
  \[
     \lim_{n \rightarrow \infty} \ \sup_{p \in Q,\sigma \in \Sigma} \  \left\{\calP_{p\bn}^\sigma [\Term \geq n^{2-\varepsilon} ] \  \right\} \quad = \quad 1
\vspace{-1mm}
  \] 
\end{theorem}

The last part of Theorem~\ref{thm-main} deserves some comments. Recall $\term_d(n) \in \Omega(n^2)$ if there is  $b>0$ such that $\term_d(n) \geq b \cdot n^2$ for all sufficiently large $n$. We prove $\term_d(n) \in \Omega(n^2)$ by showing the existence of $\delta,c > 0$ such that $\calP_{p\bn}^\sigma [\Term \geq n^2/c] \geq \delta$ for all sufficiently large $n$, where~$\sigma$ is a suitable strategy depending on $p\bn$ (then, we can put $b = \delta/c$). Hence, $\term_d(n) \in \Omega(n^2)$ does \emph{not} imply that
$\sup_{\sigma \in \Sigma, p\in Q} \  \left\{\calP_{p\bn}^\sigma [\Term \geq n^2/c ] \  \right\}$
converges to~$1$ as $n \rightarrow \infty$ (for some constant~$c$). The last part of Theorem~\ref{thm-main} shows that for an arbitrarily small $\varepsilon > 0$, we have that $\sup_{\sigma \in \Sigma,p\in Q} \  \left\{\calP_{p\bn}^\sigma [\Term \geq n^{2-\varepsilon} ] \  \right\}$ \emph{does} converge to~$1$ as $n \rightarrow \infty$. The question whether the convergence holds for $\varepsilon = 0$ remains open.

\subsection{Outline of Techniques}
\vspace{-0.2\baselineskip}
\label{subsec:outline}

The proof of Theorem~\ref{thm-main} is non-trivial, and it is based on combining the existing techniques with new analysis invented in this paper. We use VASS MDPs $\A_1$ and $\A_2$ of Fig.~\ref{fig-VASS-MDP} as running examples to illustrate our techniques.

A polynomial-time algorithm deciding asymptotic linearity of termination time for \emph{purely non-deterministic} VASS (where the set $Q_p$ is empty) was given in \cite{BCKNVZ:VASS-termination-LICS}. Theorem~\ref{thm-main} generalizes this result to VASS MDPs. We start by recalling the results of \cite{BCKNVZ:VASS-termination-LICS} and sketching the main ideas behind the proof of Theorem~\ref{thm-main}. These ideas are then elaborated in subsequent sections. 

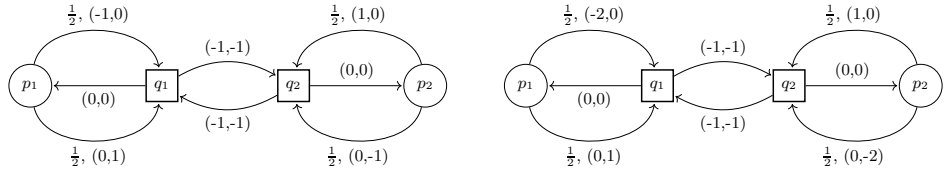
\begin{figure}[t]
	\scalebox{0.7}{%
		\begin{tikzpicture}[shorten >=1pt,node distance=2.5cm,on grid,auto]
		\tikzstyle{state2} = [rectangle,thick,draw=black,minimum size=6mm]
		
		\node[state2] (q_1)   {$q_1$};
		\node[state2] (q_2) [right=of q_1] {$q_2$};
		\node[state] (p_1) [left=of q_1] {$p_1$};
		\node[state] (p_2) [right=of q_2] {$p_2$};
		\path[->]
		(q_1) 	edge [bend left]  node {(-1,-1)} (q_2)
		edge node {(0,0)} (p_1)
		(q_2) 	edge [bend left] node  {(-1,-1)} (q_1)
		edge node {(0,0)} (p_2)
		(p_1)	edge [bend left=80,above] node {$\frac12$, (-1,0)} (q_1)
		edge [bend right=80,below] node {$\frac12$, (0,1)} (q_1)
		(p_2)	edge [bend left=80,below] node {$\frac12$, (0,-1)} (q_2)
		edge [bend right=80,above] node {$\frac12$, (1,0)} (q_2);
		\end{tikzpicture}\hspace*{3em}
		\begin{tikzpicture}[shorten >=1pt,node distance=2.5cm,on grid,auto]
\tikzstyle{state2} = [rectangle,thick,draw=black,minimum size=6mm]
\node[state2] (q_1)   {$q_1$};
\node[state2] (q_2) [right=of q_1] {$q_2$};
\node[state] (p_1) [left=of q_1] {$p_1$};
\node[state] (p_2) [right=of q_2] {$p_2$};
\path[->]
(q_1) 	edge [bend left]  node {(-1,-1)} (q_2)
edge node {(0,0)} (p_1)
(q_2) 	edge [bend left] node  {(-1,-1)} (q_1)
edge node {(0,0)} (p_2)
(p_1)	edge [bend left=80,above] node {$\frac12$, (-2,0)} (q_1)
edge [bend right=80,below] node {$\frac12$, (0,1)} (q_1)
(p_2)	edge [bend left=80,below] node {$\frac12$, (0,-2)} (q_2)
edge [bend right=80,above] node {$\frac12$, (1,0)} (q_2);
\end{tikzpicture}}
\caption{VASS MDP $\A_1$ (left) and $\A_2$ (right). The states $p_1,p_2$ are probabilistic, and the states $q_1,q_2$ are nondeterministic.}
\label{fig-VASS-MDP}
\vspace{-0.3cm}
\end{figure}

Consider a purely non-deterministic VASS $\A$ of dimension~$d$. A cycle $p_0,\bu_1,p_1,\ldots,\bu_n,p_n$ of $\A$ is \emph{simple} if all $p_1,\ldots,p_{n-1}$ are pairwise different. The total effect of a simple cycle, i.e., the sum $\sum_{i=1}^n \bu_i$, is called an \emph{increment}. Clearly, there are only finitely many increments $\bi_1,\ldots,\bi_k$. In \cite{BCKNVZ:VASS-termination-LICS}, it was shown that the termination time of $\A$ is linear iff all increments are contained in an open half-space whose normal $\bw$ is strictly positive in every component. 
The ``if'' direction is immediate, relying on a straightforward ``ranking'' argument.
The ``only if'' part is more elaborate. In \cite{BCKNVZ:VASS-termination-LICS}, it was shown that if the increments are \emph{not} contained in an open half-space with positive normal, then for all sufficiently large $n$, there is a non-terminating computation initiated in $p\bn$ whose length is at least $n^2/c$ for some constant~$c$. This computation consists of simple cycles and auxiliary short paths used to ``switch'' from one control state to another.

Now let $\A$ be a VASS MDP with $d$ counters. Here, instead of simple cycles and their increments, we use the vectors of \emph{expected counter changes per transition induced by MD strategies in their BSCCs}. More precisely, for each of the finitely many MD strategies $\sigma$ and every BSCC $\B$ of the finite-state Markov chain $\A_\sigma$ obtained by ``applying'' $\sigma$ to $\A$, we consider the unique vector $\bi$ of expected counter changes per transition (note that~$\bi$ is the same for almost all infinite computations initiated in a state of $\B$). Thus, we obtain a finite set of vectors $\bi_1,\ldots,\bi_k$ together with the associated set of tuples $(\sigma_1,\B_1), \dots, (\sigma_k,\B_k)$ where each $\sigma_i$ is an MD strategy and $\B_i$ is a BSCC of $\sigma_i$ (note that we can have $\sigma_i=\sigma_j$ for $i \neq j$ since MD strategies might have multiple BSCCs). Similarly as in \cite{BCKNVZ:VASS-termination-LICS}, we check whether all $\bi_1,\ldots,\bi_k$ are contained in an open half-space whose normal $\bw$ is strictly positive in every component. This is achievable in polynomial time by using the results of \cite{BBCFK:MDP-multiple-MP-LMCS}. If such a $\bw$ exists, we can conclude $\term_d(n) \in \bigO(n)$. This is because the ``extremal'' vectors of expected counter changes per transition are obtained by MD strategies\footnote{Here we rely on well-known results about finite-state MDPs \cite{Puterman:book}.}, and hence the expected shift in the direction opposite to $\bw$ per transition stays bounded away from zero even for general strategies. 
We than use a submartingale-based argument to show that the expected termination time is linear. This proves the first part of Theorem~\ref{thm-main}.

\begin{example}
	For the VASS MDP $\A_2$ of Fig.~\ref{fig-VASS-MDP}, there are three different increments $\bi_1 = (-1,\frac{1}{2})$, $\bi_2 = (\frac{1}{2},-1)$, and $\bi_3 = (-1,-1)$. Hence, we can choose $\bw = (1,1)$ as a positive normal satisfying $\bi_1 \cdot \bw < 0$, $\bi_2 \cdot \bw < 0$, and $\bi_3 \cdot \bw < 0$. For the VASS MDP $\A_1$ of Fig.~\ref{fig-VASS-MDP}, there are three different increments $\bi_1 = (-\frac{1}{2},\frac{1}{2})$, $\bi_2 = (\frac{1}{2},-\frac{1}{2})$, and $\bi_3=(-1,-1)$, hence no positive normal $\bw$ satisfying $\bi_1 \cdot \bw < 0$ and $\bi_2 \cdot \bw < 0$ exists.
\end{example}

Now suppose there is no such $\bw$. 
Recall that for \emph{purely non-deterministic VASS}, a sufficiently long non-terminating computation initiated in $p\bn$ consisting of simple cycles and short ``switching'' paths was constructed in \cite{BCKNVZ:VASS-termination-LICS}. Since $\bi_1,\ldots,\bi_k$ are no longer effects of simple cycles or any fixed finite executions, it is not immediately clear how to proceed and we need to use new techniques. The arguments of \cite{BCKNVZ:VASS-termination-LICS} used to construct a sufficiently long non-terminating computation are \emph{purely geometric}, and they do not depend on the fact that increments are total effects of simple cycles. Hence, by using the same construction, we obtain a sufficiently long sequence of vectors consisting of $\bi_1,\ldots,\bi_k$ and some auxiliary elements representing switches between control states. We call this sequence a \emph{scheme}, because it does not correspond to any real computation of $\A$ in general. When the constructed scheme is initiated in $p\bn$, the resulting trajectory never crosses any axis. Also note that for every fixed $r \in \N$, we can create an extra $(r-1) \cdot n$ space between the trajectory and the axes by shifting the initial point from $p\bn$ to $p(r\cdot\bn)$, which does not influence our asymptotic bounds. Now, we analyze what happens if the constructed scheme is \emph{followed} from $p(r\cdot\bn)$. Here, following a vector $\bi_j$ means to execute the transition selected by $\sigma_j$, and following a ``switch'' from $p$ to $q$ means to execute a strategy which eventually reaches $q$ with probability one (we use a strategy minimizing the expected number steps needed to reach $q$). 
Using concentration bounds of martingale theory, we show that the probability of all executions deviating from the scheme by more than $r\cdot n$ is bounded by $1-\delta$ for some fixed $\delta > 0$ (assuming $n$ is sufficiently large), which yields the $\term_d(n) \in \Omega(n^{2})$ lower bound of Theorem~\ref{thm-main}. The last part of Theorem~\ref{thm-main} is proven by a more detailed analysis of the established bounds.

Let us note that the underlying martingale analysis is not immediate, since the previous work which provides the basis for this analysis (such as \cite{BKK:pOC-time-LTL-martingale-JACM}) typically assume that the analysed strategies are memoryless in the underlying finite state space. In contrast, strategies arising of schemes are composed of multiple memoryless strategies, with the switching rules depending on the size of the initial configuration. Hence, we take a compositional approach, analysing each constituent strategy  separately using known techniques and composing the results via a new approach. 
\vspace{-0.25cm}


\subsection{The Algorithm}
\label{sec-linear}

In this section we prove the first part of Theorem~\ref{thm-main}. Our analysis uses results on \emph{multi-mean-payoff} MDPs. Recall that if $\M$ is an MDP with transitions labelled by elements of $\R^d$ (for some dimension $d$), then a \emph{mean-payoff} of an infinite path  $\pi = p_0,\bu_1,p_1,\bu_2,\ldots$ of $\A$ is $\MP(\pi)  = \liminf_{n\rightarrow \infty} \frac{1}{n} \sum_{i=1}^n \bu_i$. 
%
%
%
We say that a given vector $\bv$ is \emph{achievable} for $\A$ if there exist a strategy $\sigma$ and $p \in Q$ such that $\Exp_p^\sigma[\MP] \geq \bv$. Now we recall some results on mean-payoff MDPs \cite{BBCFK:MDP-multiple-MP-LMCS} used as tools in this section. 
 
\begin{description}
	\item[(a)] There is a finite set $\calR$ of vectors such that the set of all achievable vectors is precisely the set of all $\bv$ such that $\bv \leq \bu$ for some $\bu \in \calR^*$, where $\calR^*$ is the convex hull of $\calR$.
	\item[(b)] The problem whether a given rational $\bv$ is achievable is decidable in polynomial time.  
\end{description}

\noindent
Furthermore, we need the following result about finite-state MDPs.


\begin{lemma}
	\label{lem-MDP-finite}
	Let $\M = \ce{Q, (Q_n,Q_p),T,P}$ be a strongly connected MDP with labels from $\Q$ such that 
\vspace{-2mm}
	\[
	   \sup\, \{\Exp_p^\sigma[\MP]  \mid \sigma \in \Sigma, p\in Q\}  \ = \ \kappa \ < \ 0 
\vspace{-2mm}
	\]
    Let $\Dec$ be a function assigning to every infinite path $\pi = p_0,u_1,p_1,u_2,\ldots$ of $\M$ the least $m$ such that $\sum_{i=1}^m u_i \leq -1$. If there is no such $m$, then $\Dec(\pi) = \infty$. Then there exists a constant $c$ depending only on $\M$ such that for every $p \in Q$ and $\sigma \in \Sigma$ we have that $\Exp_p^\sigma[\Dec] \leq c$.
\end{lemma}

Now we show how to prove the first part of Theorem~\ref{thm-main} using the results above. Let $\A = \ce{Q, (Q_n,Q_p),T,P}$ be a strongly connected VASS MDP. For each of the finitely many MD strategies $\sigma$ we can consider a finite-state Markov chain $\A_\sigma$ obtained from $\A$ by fixing in every $q \in Q_n$ the probability of transitioning  to the unique successor specified by $\sigma(q)$ to 1. For each such $\A_\sigma$ and each its BSCC $\B$  we consider the unique vector $\bi$  defined by
\vspace{-2mm}
\[
    \bi \ = \ \sum_{p \in \B,\ p \tran{\bu} q} \eta(p) \cdot P(p \tran{\bu} q) \cdot \bu 
    \vspace{-2mm}
\]
where $\eta$ is the  invariant (stationary) distribution over the states of $\B$ (note that $\bi = \Exp_p^\sigma[\MP]$ for every $p \in \B$). Thus, we obtain a finite set of \emph{increments} $\bi_1,\ldots,\bi_k$ together with the associated MD strategies $\sigma_1,\ldots,\sigma_k$ and the BSCCs $\B_1,\ldots,\B_k$.

\begin{lemma}
\label{lem-kappa}
   If there exists a vector $\bw > \zv$ such that $\bi_j \cdot \bw < 0$ for every $1 \leq j \leq k$, then there exists $\kappa < 0$ such that $\bw \cdot \Exp_p^\sigma[\MP] \leq \kappa$ for every $p \in Q$ and $\sigma \in \Sigma$.  
\end{lemma}
\begin{proof}
	Let $\kappa = \max\{\bi_j \cdot \bw \mid 1 \leq j \leq k\}$.
	Consider a $\Q$-labelled MDP $\M$ obtained from $\A$ by replacing each counter update vector $\bu$ with the number $\bu \cdot \bw$. Note that every strategy $\sigma$ for $\A$ can be seen as a strategy for $\M$, and vice versa. For a given $\sigma \in \Sigma$, we write $\Exp_p^{\sigma,\A}[\MP]$ and $\Exp_p^{\sigma,\M}[\MP]$ to denote the expected value of $\MP$ in $\A$ and $\M$, respectively. Note that for every $\sigma \in \Sigma$ we have that $\Exp_p^{\sigma,\M}[\MP] = \bw \cdot \Exp_p^{\sigma,\A}[\MP]$.
	
	For every $p \in Q$, there is an optimal MD strategy $\hat{\sigma}$ maximizing the expected mean payoff in $\M$. Since  $\Exp_p^{\hat{\sigma},\M}[\MP]$ is a convex combination of increments, we obtain $\Exp_p^{\hat{\sigma},\M}[\MP] \leq \kappa$. Now let $\sigma$ be an arbitrary strategy. Since $\Exp_p^{\sigma,\M}[\MP] \leq \Exp_p^{\hat{\sigma}}[\MP] \leq \kappa$, we obtain $\Exp_p^{\sigma,\M}[\MP] = \bw \cdot \Exp_p^{\sigma,\A}[\MP] \leq \kappa$.
	\qed
\end{proof}	
A direct corollary to Lemma~\ref{lem-MDP-finite} and Lemma~\ref{lem-kappa} is the following:

\begin{lemma}
	If there exists a vector $\bw > \zv$ such that $\bi_j \cdot \bw < 0$ for every $1 \leq j \leq k$, then 
$\term_d(n) \in \mathcal{O}(n)$ holds for $\A$.
\end{lemma}

\noindent
The next lemma leads to a sound algorithm for proving of linear termination complexity.

\begin{lemma}
	The vector $\zv$ is achievable for $\A$ iff there is no $\bw > \zv$ such that $\bi_j \cdot \bw < 0$ for every $1 \leq j \leq k$.
\end{lemma}
\begin{proof}
	If $\zv$ is achievable, there exist $\sigma \in \Sigma$ and $p \in Q$ such that $\Exp_p^\sigma[\MP] \geq \zv$. Suppose there is $\bw > \zv$ such that $\bi_j \cdot \bw < 0$ for every $1 \leq j \leq k$. By Lemma~\ref{lem-kappa}, $\Exp_p^\sigma[\MP] \cdot \bw < 0$, which is a contradiction.
	
	Now suppose $\zv$ is not achievable. Consider the (convex and compact) set $\calR^*$ of claim~(a). Since $\zv$ is not achievable, the set $\calR^*$ has the empty intersection with the (convex) set of all vectors with non-negative components. By the hyperplane separation theorem, there exists a hyperplane with normal $\bw > \zv$ such that $\bv \cdot \bw < 0$ for all $\bv \in \calR^*$. Since every increment $\bi$ is achievable, there is $\bv  \in \calR^*$ such that $\bi \leq \bv$. Hence, $\bi \cdot \bw < 0$.
	\qed
\end{proof}

%
%

Hence, to check linear termination complexity, our algorithm simply checks whether $\zv$ is achievable for $\A$. The previous lemma shows that this approach is sound. In the next subsection, we show that if there is no $\bw > \zv$ such that $\bi_j \cdot \bw < 0$ for every $1 \leq j \leq k$, then the expected termination time of $\A$ is at least quadratic. This shows that our algorithm is also complete, i.e. a decision procedure for linear termination of strongly connected demonic VASS MDPs.

\vspace{-0.25cm}

\subsection{Quadratic Lower Bound}
\label{sec-scheme}

For the rest of this section, we fix a strongly connected VASS MDP $\A = \ce{Q, (Q_n,Q_p),T,P}$. Let 
$\bi_1,\ldots,\bi_k$ be the increments, and $\sigma_1,\ldots,\sigma_k$ and $\B_1,\ldots,\B_k$ the associated MD strategies and BSCCs introduced in Section~\ref{sec-linear}. 

Suppose that there does \emph{not} exist a normal vector $\bw > \zv$ such that $\bi_i \cdot \bw < 0$ for every $1 \leq i \leq k$. By~\cite[Lemma~3.2]{BCKNVZ:VASS-termination-LICS}\footnote{Technically, Lemma~3.2 in \cite{BCKNVZ:VASS-termination-LICS} assumes $\bi_j \in \Z^d$ for every $1 \leq j \leq k$. Here, $\bi_j \in \Q^d$. We can  multiply all increments of by the least common multiple of all denominators and apply Lemma~3.2 afterwards.}, there exist a subset of increments $\bj_{1},\ldots,\bj_{\ell}$ and positive integer coefficients $a_1, \ldots, a_\ell$ such that $\sum_{i=1}^\ell a_i \bj_i \geq \zv$. We use this subset to construct a so-called \emph{scheme}.

\noindent
\textbf{Scheme}
The definition of a scheme is parameterized by a certain function $L : \N \rightarrow \N$. This function is defined later, for now it suffices to know that $L(n) \in \Theta(n)$. For every $n \in \N$, we define the \emph{scheme for $n$}, which is a concatenation of $L(n)$ identical \emph{$n$-cycles}, where each $n$-cycle is defined as follows:
\vspace{-2mm}
\[
   \underbrace{\bj_1,\ldots,\bj_1}_{L(n)\cdot a_1},s_1,
   \underbrace{\bj_2,\ldots,\bj_2}_{L(n)\cdot a_2},s_2,\ \cdots \ ,
   \underbrace{\bj_\ell,\ldots,\bj_\ell}_{L(n)\cdot a_{\ell}},s_{\ell}   
\vspace{-2mm}
\]
The subsequence $\bj_i,\ldots,\bj_i,s_i$ of the $j$-th cycle is called the \emph{$i$-th segment} of the $j$-th $n$-cycle. Since the length of each $n$-cycle is $\Theta(n)$, the length of the scheme for~$n$ is $\Theta(n^2)$. 

\begin{example}
	\label{exa-scheme}
	Recall the VASS MDP $\A_1$ of Fig.~\ref{fig-VASS-MDP}. Here, we put $\bj_1 = (-\frac{1}{2},\frac{1}{2})$, $\bj_2 = (\frac{1}{2},-\frac{1}{2})$, and $a_1 = a_2 = 1$. So, the cycle for $n$ is
\vspace{-2mm}
	\[ 
	  \underbrace{\textstyle(-\frac{1}{2},\frac{1}{2}),\ldots,(-\frac{1}{2},\frac{1}{2})}_{L(n)},s_1, 
	  \underbrace{\textstyle(\frac{1}{2},-\frac{1}{2}),\ldots,(\frac{1}{2},-\frac{1}{2})}_{L(n)},s_2
\vspace{-2mm}
	\] 
\end{example}

Note that the scheme does \emph{not} necessarily correspond to any finite path in~$\A$, even if the switches are disregarded. However, the scheme for $n$ determines a unique strategy $\eta_{n}$ for $\A$ defined below.
\smallskip

\noindent
\textbf{From Schemes to Strategies.}
For every $p \in Q$, we fix an MD strategy $\gamma_p$ such that for every $q \in Q$, the $\calP_q^{\gamma_p}$ probability of visiting~$p$ from $q$ is equal to one. Furthermore, we fix some state $p_i \in \B_i$ for every $1 \leq i \leq \ell$. 

For all finite paths that are \emph{not} initiated in $p_1$, the strategy $\eta_{n}$ is defined arbitrarily. Otherwise, $\eta_{n}$ starts by simulating the strategy $\sigma_1$ for precisely $L(n)\cdot a_1$ steps. Then, $\eta_{n}$ remembers the state $q^1_1$ in which the simulation of $\sigma_1$ ended, and changes to simulating $\gamma_{p_2}$ until the state $p_2$ of $\B_2$ is reached. After reaching $p_2$, the strategy $\eta_{n}$ simulates $\sigma_2$ for precisely $L(n)\cdot a_2$ steps. Then, it again remembers the final state $q^1_2$ and starts to simulate $\gamma_{p_3}$ until $p_3$ is reached, and so on, until the simulation of $\sigma_\ell$ corresponding to the $\ell$-th segment of the first $n$-cycle is completed. Then, $\eta_{n}$ starts to simulate the switch $s_\ell$ of the first $n$-cycle, i.e., the strategy $\gamma_{q_1^1}$. This completes the simulation of the first $n$-cycle. In general, the $j$-th $n$-cycle (for $2\leq j \leq L(n)$) is simulated in the same way, the only difference is that every switch  $s_i$ is simulated by $\gamma_{q_i^{j-1}}$ where $q_i^{j-1}$ is the state entered when terminating the simulation of $\sigma_{(i+1) \mod \ell}$ in the $(j{-}1)$-th $n$-cycle. This goes on until all $n$-cycles of the scheme are simulated. After that, $\eta_{n}$ behaves arbitrarily.

\smallskip
\noindent
\textbf{Lower Bound.}
We now show that the family of strategies $\{\eta_{n}\mid n\in \N \}$ witnesses the quadratic complexity. First we define $L(n)$. From standard results on MDPs~\cite{Puterman:book} we know that for every $p$, the expected number of steps we keep playing $\gamma_p$ before hitting $p$ is finite and dependent only on $\A$. Hence, there exists a constant $\xi$ depending only on $\A$ such that also the expected change
of every counter incurred while simulating $\gamma_p$ is bounded by $\xi$.
Now let $\minup = \min\{\bu(i) \mid (p,\bu,q) \in T\}$, i.e., $\minup$ is the minimal counter update over all transitions, and let
\vspace{-3mm}
\[
L(n) = \lfloor n\, /\, (\ell \cdot \xi - \sum_{j=1}^\ell a_j \cdot \minup+1)\rfloor.
\vspace{-3mm}
\]
The function $L(n)$ has been chosen so that, for all sufficiently large $n$, if the scheme for $n$ is ``executed'' from the point $\bn$, i.e., if we follow the vectors of the scheme, where each switch is replaced with the vector $(\xi,\ldots,\xi)$, then the resulting \emph{trajectory} never crosses any axis (recall that $\sum_{i=1}^\ell a_i \bj_i \geq \zv$). 

\begin{example}
	A trajectory for the scheme of Example~\ref{exa-scheme} is shown in Fig.~\ref{exa-trajectory}. Here, $\xi = -1$, because performing every switch takes just one transition with expected change of the counters equal to $(-1,-1)$. 
\end{example}

%
%

\begin{definition}
	Let $\pi = p_0, \bu_1, p_1, \bu_2, \ldots, p_j$ be a finite  alternating sequence of states and vectors of $\Q^d$ (not necessarily a finite path in $\A$), and $m \in \N$. We say that $\pi$ is \emph{$m$-safe} if, for every $1 \leq i \leq j$, we have that $\sum_{k=1}^i \bu_k \geq - \bm$. Furthermore, we say that an infinite sequence $\pi = p_0, \bu_1, p_1, \bu_2, \ldots$ is \emph{$m$ safe-until $k$} if its prefix $p_0, \bu_1, p_1, \bu_2, \ldots, p_k$ is $m$-safe.
\end{definition}

Now consider an infinite path $\pi = q_0,\bu_1,q_1,\bu_2,\ldots$ in $\A$ initiated in~$p_1$. Then almost all such $\pi$'s (w.r.t.{} the probability measure $\calP_{p_1}^{\eta_{n}}$) can be split into a concatenation of sub-paths
\[
\pi^1_1,\tau^1_1,\ldots,\pi^1_{\ell},\tau^1_{\ell},\pi^2_1,\tau^2_1,\ldots,\pi^2_{\ell},\tau^2_{\ell}, \ldots \quad\ldots\pi^{L(n)}_1,
 \tau^{L(n)}_1,\ldots,\pi^{L(n)}_{\ell},
  \tau^{L(n)}_{\ell}, 
  \hat{\pi}
\]
where $\pi_i^{j}$ is a path with precisely $L(n) \cdot a_i$ transitions (resulting from simulation of $\sigma_i$), $\tau_i^j$ is a \emph{switching} path performing the switch $s_i$ of the $j$-th cycle, and $\hat{\pi}$ is the remaining infinite suffix of~$\pi$.
Note that for every $1 \leq i \leq \ell$, the paths $\pi^1_i, \pi^2_i,\ldots,\pi^{L(n)}_i$ can be concatenated and form a single path in $\A$ of length $L^2(n)$. This follows from the way of scheduling the switching strategies $\gamma_p$ in $\eta_{n}$. 
Writing $\pi = \varrho \odot \hat{\pi}$ (where $\hat{\pi}$ is the suffix of $\pi$ defined above), we denote by $\Sl(\pi)$ the length of $\varrho$. Note that $\Sl(\pi) \geq L^2(n)$ for almost all $\pi$. 



We now focus on proving the following lemma:

\begin{lemma}
\label{lem-lower}
   For every $\delta > 0$ there exist $r,n_0 \in \N$ such that for all $n \geq n_0$, the $\calP_{p_1}^{\eta_{n}}$ probability of all infinite paths $\pi$ initiated in $p_1$ that are $r \cdot n$-safe until $\Sl(\pi)$ is at least~$1-\delta$. Moreover, the $n_0$ is independent of $\delta$.
\end{lemma}

The lemma guarantees that if the strategy $\eta_{n}$ is executed in a configuration $p_1 (r\cdot \bn)$, where $n\geq n_0$, then $\calP_{p_1 (r \cdot \bn)}^{\eta_{n}}[\Term \geq L(n)^{2}] \geq 1-\delta$. This implies $\term(n) \in \Omega(n^2)$. Hence, it remains to prove the lemma.

\smallskip
\noindent
\textbf{{Proof of Lemma~\ref{lem-lower}}.} We separately bound the probabilities of ``large counter deviations'' while simulating the $\sigma_i$'s and the switching strategies. To this end,
for every $1 \leq i \leq \ell$ let $\pi_i = p_0,\bv_1,p_1,\bv_2,\ldots$  be the finite path of length $L^2(n)$ obtained by concatenating all $\pi^1_i, \pi^2_i,\ldots,\pi^{L(n)}_i$. Furthermore, let $\Ipath^i(\pi)$ the sequence obtained from $\pi_i$ by replacing every $\bv_k$ with $\bv_k - \bj_i$. Intuitively, $\Ipath^i(\pi)$ is $\pi_i$ where the transition effects are ``compensated'' by subtracting the expected change in the counter values per transition. 
We prove the following:
\begin{lemma}
	\label{lem-first-bound}
For 
	every $\delta > 0$, there exist $c,n_0 \in \N$ such that for all $n \geq n_0$ it holds $\calP_{p_1}^{\eta_{n}}(\{\pi\mid \Ipath^i(\pi) \text{ is $c\cdot n$-safe} \})\geq 1-\delta $. Moreover, the $n_0$ does not depend on $\delta$.
\end{lemma}
In the proof of Lemma~\ref{lem-first-bound}, we use the martingale defined for stochastic one-counter automata in \cite{BKK:pOC-time-LTL-martingale-JACM}. Intuitively, if $\Ipath^i(\pi)$ is $n$ safe, then it must be $n$ safe in every counter. Hence, we can consider each counter one by one, abstract the other counters, and estimate the probability of being $n$ safe in each of these one-counter automata. 

Similarly, we need to estimate the probability of deviating from the trajectory by performing the switches. Let $\Spath(\pi)$ be the concatenation of all $\tau_i^j$ where $1 \leq i \leq \ell$ and $1 \leq j \leq L(n)$ preserving their order. We prove the following:  
   
\begin{lemma}
	\label{lem-seconf-bound}
For 
every $\delta > 0$, there exist $c,n_0 \in \N$ such that for all $n \geq n_0$ it holds $\calP_{p_1}^{\eta_{n}}(\{\pi\mid \Spath^i(\pi) \text{ is $c\cdot n$-safe} \})\geq 1-\delta $. Moreover, the $n_0$ does not depend on $\delta$.
\end{lemma}

Clearly, if $\Ipath^i(\pi)$ is $c_1\cdot n$-safe for all $1 \leq i \leq \ell$ and $\Spath(\pi)$ is $c_2\cdot n$-safe, then $\pi$ is $(c_1+c_2) \cdot (\ell{+}1)\cdot n$-safe until $\Sl(\pi)$. Hence, Lemma~\ref{lem-lower} is a simple consequence of Lemma~\ref{lem-first-bound} and Lemma~\ref{lem-seconf-bound}. 

\smallskip
\noindent
\textbf{Probability of Quadratic Behaviour.}
Now we indicate how to prove the last part of Theorem~\ref{thm-main}. 
Directly from Lemma~\ref{lem-lower}, we have that $\lim_{r \to \infty} \calP_{p_1 (r \cdot \bn)}^{\eta_{n}} [\Term \geq L(n)^{2}] = 1$. However, observe that if $r$ is not a fixed constant, we cannot say that the size of the initial configuration is linear in~$n$. Taking $r = n^{\gamma}$ for a suitable $\gamma > 0$, we may rewrite the limit in the following way: $\lim_{r \to \infty} \calP_{p_1 (r \cdot \bn)}^{\eta_{n}} [\Term \geq L(n)^{2}] = \lim_{n \to \infty} \calP_{p_1 (\bn^{1+\gamma})}^{\eta_{n}} [\Term \geq L(n)^{2}] = \lim_{n \to \infty} \calP_{p_1 \bn}^{\eta_{n^{1/(1+\gamma)}}} [\Term \geq L(n^{1/(1+\gamma)})^{2}]$. It can be shown that $L(n^{1/(1+\gamma)})^2 > n^{2-\eps}$, for every sufficiently large $n$, thus obtaining the last part of the Theorem~\ref{thm-main}.
\vspace{-0.2cm}


\subsection{Linearity of Angelic Termination Time}
\label{sec-angel}

For angelic nondeterminism, we have a similar result as in the demonic one.

\begin{theorem}\label{thm-main-angel}
The problem whether the expected angelic termination time of a given strongly connected VASS MDP $\A$  is linear is decidable in polynomial time. If the expected angelic termination time of $\A$ is not linear, then $\term_a(n) \in \Omega(n^2)$. Furthermore, for every $\epsilon > 0$ we have that
\vspace{-1mm}
  \[
     \lim_{n \rightarrow \infty} \ \inf_{p \in Q,\sigma \in \Sigma} \  \left\{\calP_{p\bn}^\sigma [\Term \geq n^{2-\varepsilon} ] \  \right\} \quad = \quad 1
\vspace{-1mm}
  \] 
\end{theorem}

\begin{proof}[Sketch]
We analyse each counter, i.e., we consider $d$ one-dimensional VASS MDPs obtained by projecting the labelling function of $\A$.

If it is possible to terminate in one of these one-dimensional VASS MDPs in expected linear time, then the corresponding strategy achieves linear termination also in $\A$. On the other hand, if this is not possible, then every one-counter has infinite angelic termination complexity. This \emph{does not} mean that the $\A$ has infinite angelic termination complexity. However, we show that there exists a constant $c >0$ such that for sufficiently large initial configuration, the probability of runs terminating before $n^2/c$ transitions is sufficiently small for every one-counter. By union bound, the probability of runs terminating before $n^2/c$ in $\A$ is $1-\delta$ for some $\delta>0$. Thus, $\term_a(n) \in \Omega(n^2)$. The last part of the theorem is proved similarly to the demonic case.
\qed
\end{proof}

\section{General VASS MDPs \& Conclusion}
\label{sec:general}
We now drop the assumption that the VASS is strongly connected. Recall that an \emph{end-component} in an MDP is a set $\mec$ of states that is \emph{closed} (i.e., for $q\in Q_n \cap \mec$ at least one outgoing transition goes to $\mec$, while for $q\in Q_p \cap \mec$ all the outgoing transitions must end in $\mec$) and strongly connected. A \emph{maximal end component (MEC)} is an EC which is not contained in any larger EC. A decomposition of an MDP into MECs can be computed in polynomial time by standard algorithms~\cite{Alfaro:thesis}, and each MEC of a VASS MDP induces a strongly connected VASS sub-MDP which can be analyzed as shown in previous sections. We can construct a graph whose vertices correspond to MECs of an MDP and there is an edge from $\mec$ to some other $\mec'$ if and only if $\mec'$ is reachable from $\mec$. If the only cycles in this graph are self-loops, we say that the original MDP is \emph{DAG-like.} MECs corresponding to ``leafs'' of the graph (i.e. MECs that cannot be exited) are called \emph{bottom} MECs.

\begin{theorem}
\label{thm:general-main}
Theorem~\ref{thm-main} holds also for DAG-like VASS MDPs, while Theorem~\ref{thm-main-angel} holds for all VASS MDPs. In particular, a DAG-like VASS MDP $\A$ has $\term_d(n)\in\mathcal{O}(n)$ if and only if each MEC of $\A$ induces a (strongly connected) VASS MDP in which $\term_d(n)\in\mathcal{O}(n)$; and $\A$ has $\term_a(n)\in\mathcal{O}(n)$ iff each bottom MEC of $\A$ has $\term_a(n)\in\mathcal{O}(n)$. Otherwise, the termination complexity of $\A$ is in $\Omega(n^2)$. 
\end{theorem}
\begin{proof}[Sketch]
We sketch the proof for the demonic case where there are no self-loops in the MEC graph. Then no MEC can be re-entered once left. Moreover, there is a constant $c$ s.t. whenever we enter a MEC with a counter valuation $\bv$, the expected time to either terminate or exit the MEC, as well as the expected size of the counter valuation at the time of termination/exiting are bounded by $c\cdot \size{\bv}$. Hence, a straightforward induction on the number of MECs shows that the expected maximal counter value as well as the expected termination time are bounded by $c^{|Q|}\cdot n$ from any initial configuration of size $n$. Since $|Q|$ does not depend on $n$, we get the result.
%
\end{proof}

\begin{figure}[t]
	\begin{minipage}[b]{0.45\linewidth}
		\begin{center}
			\scalebox{0.18}{
				\begin{tikzpicture}\tikzstyle{tr} = [very thick,->,>=stealth, big arrow]
				\tkzInit[xmax=22,ymax=22,xmin=-2,ymin=-2]
				\begin{scope}[very thin] 
				\draw[-{Latex[length=3mm]},color = gray] (0,0) -- (21,0);
				\draw[-{Latex[length=3mm]},color = gray] (0,0) -- (0,23);
				\end{scope}
				\draw[very thin] (0,0) grid (20,22);
				\foreach \x in {0,1,2,3,4,5,6,7,8,9,10}{
					\draw (\x+\x,0) node[below,scale=3]{\large $\x$};
					\draw (0,\x+\x) node[left,scale=3]{\large $\x$};
				}
				
				\foreach \k in {0,4,8}{
					\draw[tr] (11-\k,21-\k) -- (9-\k,19-\k);
					\draw[tr] (14-\k,14-\k) -- (12-\k,12-\k);   
					\foreach \n in {0,...,4}{
						\draw[tr] (16-\k-\n,16-\k+\n) -- (15-\k-\n,17-\k+\n);
						\draw[tr] (13-\k-\n,15-\k+\n) -- (14-\k-\n,14-\k+\n);   
				}}
				
				\end{tikzpicture}}
			\vspace{-0.2cm}
		\end{center}
		\caption{A trajectory for the scheme of Example~\ref{exa-scheme}.}
		\label{exa-trajectory}
	\end{minipage}
\hfill
	\begin{minipage}[b]{0.45\linewidth}

	\scalebox{0.8}{%
		\begin{tikzpicture}[shorten >=1pt,node distance=2.5cm,on grid,auto]
		\tikzstyle{state2} = [rectangle,thick,draw=black,minimum size=6mm]
		\node[state2] (p_1) {$p_1$};
		\node[state, below right = 1cm and 1.5 cm of p_1] (r) {$r$};
		\node[state2, right = 3cm of p_1] (p_2) {$p_2$};
		\path[->] (p_1) 	edge [bend left]  node {(0,0)} (p_2);
		\path[->] (p_2) edge [bend left]  node {(0,0)} (r); 
		\path[->] (r) edge [bend left]  node {$\frac{1}{4}$,(0,0)} (p_1);
		\node[state2, below = 1.5cm of r] (f) {$f$};
		\path[->] (r) edge  node {$\frac{3}{4}$,(0,0)} (f);
		\draw[->, loop left, looseness = 10] (p_1) to node {$(2,-1)$} (p_1);
		\draw[->, loop right, looseness = 10] (p_2) to node {$(-1,2)$} (p_2);
		\draw[->, loop left, looseness = 10] (f) to node {$(0,-1)$} (f);
		\end{tikzpicture}
%
%
}
	\caption{VASS MDP with linear MECs but infinite expected termination time.}
	\label{fig:nonscc}

\end{minipage}

\vspace{-0.3cm}
\end{figure}

For non-DAG-like VASS MDPs, the situation gets much more complicated. Consider the MDP in Figure~\ref{fig:nonscc}. There are three MECs, each a singleton ($\{p_1\}$, $\{p_2\}$, $\{f\}$). Clearly all these three MECs have a linear termination complexity. Now consider the following demonic strategy starting in configuration $p_1(0,n)$: select the loop until we get the configuration $p_1(2n,0)$; then transition to $p_2$ and play its loop until we get into $p_2(0,4n)$; then transition to $r$ and if the randomness takes us back to $p_1$, play the loop again until we get $p_1(8n,0)$, etc. \emph{ad infinitum}. Clearly, the strategy eventually ends up in $f$ where it terminates. However, the expected termination time is at least $\frac{3}{4}\sum_{i=0}^{\infty}(\frac{1}{4})^i \cdot 4^{i+1} = 3\sum_{i=0}^{\infty}(\frac{4}{4})^i = \infty.$ 

Hence, proving the linear termination complexity in general VASS does not reduce to analysing individual MECs. Moreover, it crucially depends on the concrete probabilities in transient (non-MEC) states: in Figure~\ref{fig:nonscc}, the termination time would be finite (and linear) if the transition from $r$ to $f$ had probability $<\frac{1}{4}$. The transient behaviour of MDPs can be of course rather complex and it is not even clear whether the linear demonic termination complexity is even \emph{decidable} for VASS MDPs with general structure. We see this as a very intriguing, yet complex, direction for future work.


\bibliography{concur,str-long,PL,NewCite}

\begin{thebibliography}{10}
\providecommand{\url}[1]{\texttt{#1}}
\providecommand{\urlprefix}{URL }
\providecommand{\doi}[1]{https://doi.org/#1}

\bibitem{Alfaro:thesis}
de~Alfaro, L.: {Formal verification of probabilistic systems}. Phd. thesis,
  Stanford University, Stanford, CA, USA (1998)

\bibitem{conf/icalp/AminofRZS15}
Aminof, B., Rubin, S., Zuleger, F., Spegni, F.: Liveness of parameterized timed
  networks. In: Proceedings of {ICALP} 2015. pp. 375--387 (2015)

\bibitem{AH11:Yen}
Atig, M.F., Habermehl, P.: On yen's path logic for petri nets. International
  Journal of Foundations of Computer Science  \textbf{22}(04),  783--799 (2011)

\bibitem{BK:book}
Baier, C., Katoen, J.P.: Principles of Model Checking (2008)

\bibitem{BartheGGHS:16-diff-privacy-prob-couplings}
Barthe, G., Gaboardi, M., Gr{\'e}goire, B., Hsu, J., Strub, P.Y.: Proving
  differential privacy via probabilistic couplings. In: Proceedings of LICS'16.
  pp. 749--758. ACM, New York, NY, USA (2016)

\bibitem{Bloem16}
Bloem, R., Jacobs, S., Khalimov, A., Konnov, I., Rubin, S., Veith, H., Widder,
  J.: Decidability in parameterized verification. SIGACT News  \textbf{47}(2),
  53--64 (2016)

\bibitem{BG11:Vass}
Bozzelli, L., Ganty, P.: Complexity analysis of the backward coverability
  algorithm for vass. In: Proceedings of {RP} 2011. pp. 96--109 (2011)

\bibitem{BBCFK:MDP-multiple-MP-LMCS}
Br\'{a}zdil, T., Bro{\v{z}}ek, V., Chatterjee, K., Forejt, V., Ku{\v{c}}era,
  A.: Markov decision processes with multiple long-run average objectives
  \textbf{10}(1),  1--29 (2014)

\bibitem{BBEK:OC-games-termination-approx}
Br\'{a}zdil, T., Bro{\v{z}}ek, V., Etessami, K., Ku\v{c}era, A.: Approximating
  the termination value of one-counter {MDPs} and stochastic games. In:
  Proceedings of ICALP 2011, Part II. vol.~6756, pp. 332--343 (2011)

\bibitem{BCKNVZ:VASS-termination-LICS}
Br\'{a}zdil, T., Chatterjee, K., Ku{\v{c}}era, A., Novotn{\'{y}}, P., Velan,
  D., Zuleger, F.: Efficient algorithms for asymptotic bounds on termination
  time in {VASS}. In: Proceedings of LICS 2018. pp. 185--194 (2018)

\bibitem{BKK:pOC-time-LTL-martingale-JACM}
Br\'{a}zdil, T., Kiefer, S., Ku{\v{c}}era, A.: Efficient analysis of
  probabilistic programs with an unbounded counter. J.~ACM  \textbf{61}(6)
  (2014)

\bibitem{BKNW:OC-MDP-term-time}
Br\'{a}zdil, T., Ku{\v{c}}era, A., Novotn\'{y}, P., Wojtczak, D.: Minimizing
  expected termination time in one-counter {Markov} decision processes. In:
  Proceedings of ICALP 2012, Part II. vol.~7392, pp. 141--152 (2012)

\bibitem{Cassez:11-wcet-survey}
{Cassez}, F.: Timed games for computing {WCET} for pipelined processors with
  caches. In: 2011 Eleventh International Conference on Application of
  Concurrency to System Design. pp. 195--204 (June 2011)

\bibitem{CDHR10}
Chatterjee, K., Doyen, L., Henzinger, T.A., Raskin, J.F.: Generalized
  mean-payoff and energy games. In: Proceedings of {FSTTCS} 2010. pp. 505--516
  (2010)

\bibitem{CV13}
Chatterjee, K., Velner, Y.: Hyperplane separation technique for
  multidimensional mean-payoff games. In: Proceedings of {CONCUR} 2013. pp.
  500--515 (2013)

\bibitem{CFG17}
Chatterjee, K., Fu, H., Goharshady, A.K.: Non-polynomial worst-case analysis of
  recursive programs. In: CAV (2017)

\bibitem{ChatterjeeFuMurhekar:17-automated-recurrence}
Chatterjee, K., Fu, H., Murhekar, A.: Automated recurrence analysis for
  almost-linear expected-runtime bounds. In: Computer Aided Verification.
  Springer International Publishing (2017)

\bibitem{CJLS17}
Colcombet, T., Jurdzinski, M., Lazic, R., Schmitz, S.: Perfect half space
  games. In: Proceedings of {LICS} 2017. pp. 1--11 (2017)

\bibitem{DBLP:conf/stoc/CzerwinskiLLLM19}
Czerwinski, W., Lasota, S., Lazic, R., Leroux, J., Mazowiecki, F.: The
  reachability problem for petri nets is not elementary. In: Proceedings of
  {STOC} 2019. pp. 24--33 (2019)

\bibitem{Esparza:PN}
Esparza, J.: Decidability and complexity of petri net problems -- an
  introduction. Lectures on Petri nets I: Basic models pp. 374--428 (1998)

\bibitem{EKM:prob-PDA-PCTL-LMCS}
Esparza, J., Ku{\v{c}}era, A., Mayr, R.: Model-checking probabilistic pushdown
  automata  \textbf{2}(1:2),  1--31 (2006)

\bibitem{ELMMN14:SMT-coverability}
Esparza, J., Ledesma-Garza, R., Majumdar, R., Meyer, P., Niksic, F.: An
  smt-based approach to coverability analysis. In: Proceedings of {CAV} 2014.
  pp. 603--619 (2014)

\bibitem{EN94}
Esparza, J., Nielsen, M.: Decidability issues for petri nets -- a survey.
  Bulletin of the EATCS  \textbf{52},  245--262 (1994)

\bibitem{EY:RMC-LTL-complexity-TCL}
Etessami, K., Yannakakis, M.: Model checking of recursive probabilistic systems
   \textbf{13} (2012)

\bibitem{FosterKMRS:16-prob-netkat}
Foster, N., Kozen, D., Mamouras, K., Reitblatt, M., Silva, A.: Probabilistic
  netkat. In: Thiemann, P. (ed.) European Symposium on Programming. LNCS,
  vol.~9632, pp. 282--309. Springer Berlin Heidelberg, Berlin, Heidelberg
  (2016)

\bibitem{Ghahramani:15-prob-ai-nature}
Ghahramani, Z.: Probabilistic machine learning and artificial intelligence.
  Nature  \textbf{521}(7553),  452--459 (2015)

\bibitem{GulwaniMehraChulimbi:09-speed}
Gulwani, S., Mehra, K.K., Chilimbi, T.: Speed: Precise and efficient static
  estimation of program computational complexity. In: Proceedings of POPL'09.
  pp. 127--139. ACM, New York, NY, USA (2009)

\bibitem{journals/toplas/0002AH12}
Hoffmann, J., Aehlig, K., Hofmann, M.: Multivariate amortized resource
  analysis. {ACM} Trans. Program. Lang. Syst.  \textbf{34}(3),  14:1--14:62
  (2012)

\bibitem{JLS15}
Jurdzinski, M., Lazic, R., Schmitz, S.: Fixed-dimensional energy games are in
  pseudo-polynomial time. In: Proceedings of {ICALP} 2015. pp. 260--272 (2015)

\bibitem{KaminskiKatoenMateja:19-pp-hardness-act-inf}
Kaminski, B.L., Katoen, J., Matheja, C.: On the hardness of analyzing
  probabilistic programs. Acta Inf.  \textbf{56}(3),  255--285 (2019)

\bibitem{KaminskiKMO:18-wp-expected-runtime-jacm}
Kaminski, B.L., Katoen, J., Matheja, C., Olmedo, F.: Weakest precondition
  reasoning for expected runtimes of randomized algorithms. J. {ACM}
  \textbf{65}(5),  30:1--30:68 (2018)

\bibitem{KM69}
Karp, R.M., Miller, R.E.: Parallel program schemata. J. Comput. Syst. Sci.
  \textbf{3}(2),  147--195 (1969)

\bibitem{Kosaraju82:VASS-reach-dec}
Kosaraju, S.R.: Decidability of reachability in vector addition systems
  (preliminary version). In: Proceedings of {STOC} 1982. pp. 267--281. ACM
  (1982)

\bibitem{KS88}
Kosaraju, S.R., Sullivan, G.F.: Detecting cycles in dynamic graphs in
  polynomial time. In: Proceedings of {STOC} 1988. pp. 398--406 (1988)

\bibitem{Leroux:VASS-reachability-short}
Leroux, J.: Vector addition system reachability problem: A short self-contained
  proof. In: Proceedings of POPL 2011. pp. 307--316 (2011)

\bibitem{Leroux:VASS-polynomial-term-comp}
Leroux, J.: Polynomial vector addition systems with states. In: Proceedings of
  ICALP 2018. vol.~107, pp. 134:1--134:13 (2018)

\bibitem{DBLP:journals/corr/abs-1903-08575}
Leroux, J., Schmitz, S.: Reachability in vector addition systems is
  primitive-recursive in fixed dimension. In: Proceedings of {LICS} 2019 (2019)

\bibitem{Lipton:PN-Reachability}
Lipton, R.: The reachability problem requires exponential space. Technical
  report~62 (1976)

\bibitem{Mayr:PN-reachability}
Mayr, E.: An algorithm for the general {P}etri net reachability problem
  \textbf{13},  441--460 (1984)

\bibitem{NgoCarbonneauxHoffmann:18-pp-resource-analysis}
Ngo, V.C., Carbonneaux, Q., Hoffmann, J.: Bounded expectations: Resource
  analysis for probabilistic programs. In: Proceedings of PLDI'18. pp.
  496--512. ACM, New York, NY, USA (2018)

\bibitem{Puterman:book}
Puterman, M.: {Markov} Decision Processes (1994)

\bibitem{Rackoff:Covering-TCS}
Rackoff, C.: The covering and boundedness problems for vector addition systems.
  Theor. Comput. Sci.  \textbf{6},  223--231 (1978)

\bibitem{Schmitz16:hyperackermannian-complexity-hierarchy}
Schmitz, S.: Complexity hierarchies beyond elementary. ACM Trans. Comput.
  Theory  \textbf{8}(1),  3:1--3:36 (Feb 2016)

\bibitem{SZV14}
Sinn, M., Zuleger, F., Veith, H.: A simple and scalable static analysis for
  bound analysis and amortized complexity analysis. In: Proccedings of CAV
  2014. pp. 745--761 (2014)

\bibitem{ThrunBurgardFox:2005-robotics}
Thrun, S., Burgard, W., Fox, D.: Probabilistic Robotics (Intelligent Robotics
  and Autonomous Agents). The MIT Press (2005)

\bibitem{VCDHRR15}
Velner, Y., Chatterjee, K., Doyen, L., Henzinger, T.A., Rabinovich, A.M.,
  Raskin, J.: The complexity of multi-mean-payoff and multi-energy games. Inf.
  Comput.  \textbf{241},  177--196 (2015)

\bibitem{Wilhelm&al:2008:wcet-survey}
Wilhelm, R., Engblom, J., Ermedahl, A., Holsti, N., Thesing, S., Whalley, D.,
  Bernat, G., Ferdinand, C., Heckmann, R., Mitra, T., Mueller, F., Puaut, I.,
  Puschner, P., Staschulat, J., Stenstr\"{o}m, P.: The worst-case
  execution-time problem\&mdash;overview of methods and survey of tools. ACM
  Trans. Embed. Comput. Syst.  \textbf{7}(3),  36:1--36:53 (2008)

\bibitem{Yen92:Petri-Net-logic}
Yen, H.C.: A unified approach for deciding the existence of certain petri net
  paths. Inf. Comput.  \textbf{96}(1),  119--137 (1992)

\end{thebibliography}

\appendix
\newpage
\section{Proofs}
\label{app-proofs}

We start by recalling basic notions of martingale theory. A stochastic process $m^{(0)},m^{(1)},m^{(2)},\ldots$ is a \emph{martingale} if the following holds for all $i \in \N$:

\begin{itemize}
	\item $\Exp[m^{(i)}] < \infty$,
	\item $\Exp[m^{(i+1)} \mid m^{(i)},\ldots m^{(0)}] = m^{(i)}$. 
\end{itemize}

\noindent
By weakening the second condition into $\Exp[m^{(i+1)} \mid m^{(i)},\ldots m^{(0)}] \geq m^{(i)}$, we obtain a \emph{submartingale}. 

If $m^{(0)},m^{(1)},m^{(2)},\ldots$ is a (sub)martingale such that $|m^{(i+1)} - m^{(i)}| \leq d$ almost surely for all $i \in \N$, then the Azuma-Hoeffding inequality says that, for every $t>0$,

\begin{itemize}
	\item $\calP[m^{(i)} - m^{(0)} \geq t] \ \leq \ \exp(-t^2/2id^2)$\hspace*{1cm} if $m^{(0)},m^{(1)},m^{(2)},\ldots$ is a martingale,
	\item $\calP[m^{(i)} - m^{(0)} \leq -t] \ \leq \ \exp(-t^2/2id^2)$\hspace*{.75cm} if $m^{(0)},m^{(1)},m^{(2)},\ldots$ is a submartingale.
	
\end{itemize}

\subsection{A proof of Lemma~\ref{lem-MDP-finite}}
\label{lemMDP-fin-proof}

We start by recalling the results of \cite{Puterman:book}. Let $\calN = \ce{Q, (Q_n,Q_p),T,P}$ be a strongly connected $\Q$-labeled  MDP. Consider the following linear program:
\begin{align*}
\textbf{minimize } &x \text{ subject to}\\
z_q &\geq -x + c + z_p & \text{for all } q &\in Q_n \text{ and } (q,c,p) \in T\\
\textstyle z_q &\geq -x + \sum_{(q,c,p) \in T} P(q,c,p) \cdot (c + z_p) & \text{for all } q &\in Q_p 
\end{align*}

This linear program is feasible, and the minimal value of $x$ is equal to
\[
   \sup\, \{\Exp_p^\sigma[\MP]  \mid \sigma \in \Sigma, p\in Q\}
\]
Let $\bar{x},\bar{z}_p$ be the components of an optimal solution, and let $p \in Q$ be some fixed initial state. For every $i \in \N$, let  $S^{(i)}$ and $C^{(i)}$ be functions assigning to every infinite path $\pi = p_0,\bu_1,p_1,\bu_2,\ldots$ initiated in $p$ the state $p_i$ and the sum $\sum_{j=1}^i \bu_j$, respectively (we put $C^{(0)}(\pi) = 0$). Let $\sigma \in \Sigma$ be an arbitrary strategy.
Almost identical computation as in~\cite{BBEK:OC-games-termination-approx,BKNW:OC-MDP-term-time} gives that the stochastic process $m^{(0)},m^{(1)},\ldots$, where $m^{(i)} = C^{(i)} + \bar{z}_{S^{(i)}} - i\cdot \bar{x}$, is a supermartingale (over the probability space determined by $\sigma$). 

Our aim is to show that there exists $a \in (0,1)$ and $i_0 \in \N$ depending only on $\M$ such that for an arbitrary strategy $\sigma$, every $p \in Q$, and all $i \geq i_0$ we have that $\calP_{p}^{\sigma}[\Dec = i] \leq a^i$. From this we immediately obtain 
\[
 \Exp_p^\sigma[\Dec] \ =  \ \sum_{i=1}^\infty i \cdot \calP_{p}^{\sigma}[\Dec = i] \ \leq \ c
\]
 for some constant $c$ depending only on $\M$.

Now consider the supermartingale of the first paragraph applied to infinite paths in $\M$ (under the strategy $\sigma$). Let $\pi = p_0,\bu_1,p_1,\bu_2,\ldots$ be an infinite path such that $\Dec(\pi) = i$. Then $C^{(i)}(\pi) \geq -(1 + \delta)$ for some fixed $\delta$ depending only on $\M$, because the rewards are bounded. Furthermore, the maximal difference between $\bar{z}_p$ and $\bar{z}_q$ for $p,q \in Q$ is also bounded by some constant depending only on $\M$. Hence,  
$m^{(i)}(\pi) - m^{(0)}(\pi) \geq  \mu - i \cdot \bar{x}$ for some constant $\mu$ depending only on $\M$ (recall that $\bar{x}$ is \emph{negative}). For every $i \geq i_0$ where $i_0 := -2\mu/\bar{x}$, we obtain $\mu - i \cdot \bar{x} \leq i \cdot \bar{x}/2 - i \cdot \bar{x} = - (i \cdot \bar{x}/2)$ where $i \cdot \bar{x}/2 < 0$. Thus, we obtain
\[
   \calP_{p}^{\sigma}[\Dec = i] \ \leq \ \calP_{p}^{\sigma}[m^{(i)} - m^{(0)} \geq \mu - i \cdot \bar{x}] \ \leq \ \calP_{p}^{\sigma}[m^{(i)} - m^{(0)} \geq - (i \cdot \bar{x}/2) ]
\]
Since the supermartingale $m^{(0)},m^{(1)},\ldots$ can change in one step at most by a constant $\varrho$ (depending only on $\M$), applying Azuma's inequality yields
\[
   \calP_{p}^{\sigma}[m^{(i)} - m^{(0)} \geq - (i \cdot \bar{x}/2) ] \ \leq \ \exp \left(\frac{-i^2\cdot \bar{x}^2}{8 \cdot i \cdot \varrho^2}\right) \ = \ a^i
\] 
where $a = \exp(-\bar{x}^2/8\varrho^2) \in (0,1)$, for all~$i \geq i_0$.

\subsection{A proof of Lemma~\ref{lem-first-bound}}


We assume a fixed $1 \leq i \leq \ell$. Recall the strategy $\sigma_i$ and the BSCC $\B_i$ associated to the increment $\bj_i$ (see Section~\ref{sec-linear}).  For every path $\pi$ initiated in $p_1$, let $\Ipath^i(\pi) = q_0,\bv_1,q_1,\bv_2, \dots, q_{L^2(n)}$, and for every $1 \leq j \leq d$, consider the sequence $\Ipath^i_j(\pi)$ obtained by projecting every $\bv_k$ to its $j$-th component, i.e., $\Ipath^i_j(\pi) = q_0,\bv_1(j),q_1,\bv_2(j),\dots, q_{L^2(n)}$. 

Our aim is to show that, for every $\delta > 0$ and every $1 \leq j \leq d$, there exist $c,n_0 \in \N$ such that the $\calP_{p_1}^{\sigma_i}$ probability of all $\pi$ initiated in $p_1$ such that $\Ipath^i_j(\pi)$ is $c \cdot n$ safe is at least $1-\delta$. Observe that Lemma~\ref{lem-first-bound} is a direct consequence of this claim.

Observe that $\B_i$, where the nondeterministic choice is resolved by $\sigma_i$, and the counter update vectors are projected to their $j$th component, can be seen as a one-counter automaton. The long-run average change of the counter per transition in this automaton is $\bj_i(j)$. Recall that $\Ipath^i(\pi)$ was obtained from the concatenated paths $\pi^1_i, \pi^2_i,\ldots,\pi^{L(n)}_i$ by subtracting the increment $\bj_i$ from each vector occurring in this sequence. Hence, we need to subtract $\bj_i(j)$ from every counter update on every transition of $\B_i$. Thus, we obtain a one-counter automaton $\hat{\B}_i$. A trivial but crucial observation is that the long-run average change of the counter per transition in $\hat{\B}_i$ is \emph{zero}.

The $\calP_{p_1}^{\sigma_i}$ probability of all $\pi$ initiated in $p_1$ such that $\Ipath^i_j(\pi)$ is \emph{not} $c \cdot n$ safe is equal to the probability that a run of $\hat{\B}_i$ initiated in $q(0)$, where $q$ is the starting state of $\pi^1_i$, decreases the counter to $-c\cdot n$ or below during the first $L^2(n)$ transitions. An upper bound on the latter probability can be established using the martingale for probabilistic one-counter automata introduced in \cite{BKK:pOC-time-LTL-martingale-JACM} (here we use a slightly modified version of this martingale which better suits our purposes). Due to \cite{BKK:pOC-time-LTL-martingale-JACM}, for every state $p^{(0)}$ of $\hat{\B}_i$ and every $c \in \N$, there exists a vector $\by \in [0,\infty)^{|\hat{\B}_i|}$ such that the stochastic process defined by
\[
	m^{(k)}_{i,j} = \begin{cases}
		C^{(k)} + \by_j(p^{(k)})		& \text{if } C^{(k)} \geq -cn \text{ for all } 0 \leq k' < k;\\
		m^{(k-1)}_{i,j} 				& \text{otherwise}
	\end{cases}
\]
is a martingale, where $C^{(0)} = 0$, $C^{(k)} = \sum_{s=1}^k \bv_s(j)$ is a random variable returning the accumulated counter change after $k$ steps, and $p^{(k)}$ is a random variable returning the control state entered after $k$ steps. Moreover, the vector $\by_j$ satisfies $0 \leq \|\by_j\| \leq 2|\hat{\B}_i|/x_{\min}^{|\hat{\B}_i|}$, where $x_{\min}$ is the minimum probability used in the transitions of $\hat{\B}_i$.

Note that if the accumulated counter change drops to $-c\cdot n$ or below for the first time after exactly $k$ transitions, the martingale does not change its value from this point on, and remains equal to $m^{(k)}_{i,j}$. Let $\by = \max_j \|\by_j\|$. If $m^{(L^2(n))}_{i,j} \geq -cn + \|\by\|$ then the value $C^{(L^2(n))}$ is at least $-cn$ for every $k \leq L^2(n)$. Hence, the probability that a run of $\hat{\B}_i$ initiated in $p_i(0)$ decreases the counter to $-cn$ or below during the first $L^2(n)$ transitions is less or equal to $P(m^{(L^2(n))}_{i,j} \leq -cn +  \|\by\| + 1)$. Furthermore, for all $n \geq \|\by\| +1 + m^{(0)}_{i,j}$ we have that
\[
P(m^{(L^2(n))}_{i,j} \leq -cn +  \|\by\| + 1)  \ \leq \ 
P(m^{(L^2(n))}_{i,j} - m^{(0)}_{i,j} \leq -(c-1)n) \ \leq \ \mathrm{exp}\left( -(c-1)^2/ \alpha \right)
\]
by applying Azuma's inequality, where $\alpha$ is a suitable constant dependent only on $\hat{\B}_i$.

The above holds for every $j \in \{1,\ldots,d\}$. Hence, the probability of all $\calP_{p_1}^{\sigma_i}$ probability of all $\pi$ initiated in $p_1$ such that $\Ipath^i(\pi)$ is \emph{not} $c \cdot n$ safe is less or equal to $d \cdot \mathrm{exp}\left( -(c-1)^2/ \alpha \right)$ for every sufficiently large $n$. To achieve  $d \cdot \mathrm{exp}\left( - (c-1)^2/\alpha\right) \leq \delta$, we can put
$c = \lceil \sqrt{\alpha(\ln d - \ln \delta)} \rceil + 1$.

%
%
%
%

\subsection{A proof of Lemma~\ref{lem-seconf-bound}}

First, we bound the expected number of transitions used in executing one switch. Let $x_{\min}$ be the minimum probability appearing in the VASS MDP, and let $p,q \in Q$. There is a path of length at most $|Q|-1$ which we may follow with probability at least $x_{\min}^{|Q|-1}$. If successful, we are done, otherwise we end up in some state $p'$ and again there is some path from $p'$ to $q$ of length at most $|Q|-1$ and the probability of traversing this path is still at least $x_{\min}^{|Q|-1}$. 

If we use the number $x_{\min}^{|Q|-1}$ as the (lower bound on the) probability of success, the random variable counting the number of attempts until the first success has a geometric distribution. Its expected value is then $1/x_{\min}^{(|Q|-1)}$.
Since every attempt uses at most $|Q|-1$ transitions, the expected number of used transitions is bounded from above by $\lambda = (|Q|-1) \cdot 1/x_{\min}^{(|Q|-1)}$.

Let $X$ be the random variable equal to the length of $\Spath(\pi)$. Since the number of switches is $L(n)$, the expected value of $\Spath(\pi)$ is bounded from above by $\lambda \cdot L(n)$. Let $\delta > 0$. Surely
$P(X \geq \delta^{-1} \cdot \lambda \cdot L(n)) \leq P(X \geq \delta^{-1} \Exp[X])$. By Markov inequality, we obtain that $P(X \geq \delta^{-1} \Exp[X]) \leq \delta$. Therefore, with probability at most $\delta$, we use more than $\delta^{-1} \cdot \lambda L(n)$ transitions.

The minimal update over all transitions is $\minup$. If $\minup \geq 0$, then $c = 0$ since no transition can decrease the counters and the set of paths $\pi$ such that $\Spath(\pi)$ is 0 safe has probability one. For $\minup < 0$, the set of paths $\pi$ such that $\Spath(\pi)$ is $-\minup \delta^{-1} \lambda \cdot L(n)$ safe has probability at least $1-\delta$. Since $L(n) \leq n$, taking $c = -\minup \delta^{-1} \lambda$ completes the proof.

\subsection{A proof of the last part of Theorem~\ref{thm-main}}
\label{app-last}
We show that for every $\eps>0$ we can choose $\gamma >0$ such that
\[
 \lim_{n \to \infty} \calP_{p\bn}^{\eta_{\bn^{1/1+\gamma}}} [\Term \geq n^{2-\eps}] = 1.
\] 

From Lemma~\ref{lem-lower}, we have that $\lim_{r \to \infty} \calP_{p_1 (r \cdot \bn)}^{\eta_{n}} [\Term \geq L(n)^{2}] = 1$.
Rewriting the limit, we obtain
\begin{align*}
\lim_{r \to \infty} \calP_{p_1 (r \cdot \bn)}^{\eta_{n}} [\Term \geq L(n)^{2}] &= \lim_{n \to \infty} \calP_{p_1 (\bn^{1+\gamma})}^{\eta_{n}} [\Term \geq L(n)^{2}] \\
&= \lim_{n \to \infty} \calP_{p_1 \bn}^{\eta_{n^{1/(1+\gamma)}}} [\Term \geq L(n^{1/(1+\gamma)})^{2}].
\end{align*}
Let $c = \ell \cdot \xi - \sum_{j=1}^\ell a_j \cdot \minup+2$. Then $L(n) \geq n/c$ by the definition of $L(n)$ (for $n$ sufficiently large). We need to show that
$L(n^{1/(1+\gamma)})^{2} \geq n^{2-\eps}$ for some $\gamma > 0$. Surely
\[
L(n^{1/(1+\gamma)})^{2} \geq n^{2/(1+\gamma)}/c = n^{2 - 2\gamma/(1+\gamma)}/c.
\]
For  $n$ sufficiently large, we have $n^{\gamma/(1+\gamma)}/c>1$. Therefore,
\[
n^{2 - 2\gamma/(1+\gamma)}/c \geq n^{2-3\gamma/(1+\gamma)}.
\]
Let $\gamma$ be such that $n^{2-3\gamma/(1+\gamma)} = n^{2-\eps}$, therefore $3\gamma/(1+\gamma) = \eps$. Multiplying by $1+\gamma$ we get $\eps + \eps\gamma - 3 \gamma = 0$. Therefore, $\gamma = \eps/(3-\eps)$. This completes the proof.

\subsection{A proof of Theorem~\ref{thm-main-angel}}
The proof is very similar to the one of Lemma~\ref{lem-MDP-finite}. Again, we recall the results of~\cite{BKNW:OC-MDP-term-time} for one-counter machines. We consider the following linear program:
\begin{align*}
\textbf{maximize } &x \text{ subject to}\\
z_q &\leq -x + c + z_p & \text{for all } q &\in Q_n \text{ and } (q,c,p) \in T\\
\textstyle z_q &\leq -x + \sum_{(q,c,p) \in T} P(q,c,p) \cdot (c + z_p) & \text{for all } q &\in Q_p 
\end{align*}
This linear program is feasible, and the maximal value of $x$ is equal to
\[
   \inf\, \{\Exp_p^\sigma[\MP]  \mid \sigma \in \Sigma, p\in Q\}.
\]
Moreover, we can assume that for all $q \in Q$ we have $\bar{z}_q \geq 0$. Direct corollary of~\cite[Proposition 5,(B)]{BKNW:OC-MDP-term-time} is the following Lemma:
\begin{lemma}
Let $(\bar{x}, (\bar{z}_q)_{q \in Q})$ be a solution of the linear program above. If $\bar{x} < 0$ then $\term_{a}(n) \in \Theta(n)$.
\end{lemma}

Moreover, we use (similarly to the proof of Lemma~\ref{lem-MDP-finite}) that the stochastic process $m^{(0)},m^{(1)},\dots,$ where 
\[
m^{(i)} =
\begin{cases}
C^{(i)} + \bar{z}_{S^{(i)}} - i\cdot \bar{x} & C^{(j)} > 0 \text{ for all j}, 0\leq j < i,\\
m^{(i-1)} & \text{otherwise}
\end{cases}
\]
is a submartingale.\\

Now we use these results on one-dimensional VASS-MDPs to obtain the proof for $d$-dimensional VASS-MDP $\A$ by simply considering $d$ projections on one counter.

A trivial observation gives the following result.
\begin{lemma}
Let $\A$ be a $d$-dimensional VASS-MDP, $\A_1,\dots,\A_d$ corresponding one-dimensional VASS-MDPs obtained by projecting the labels onto respective coordinate. If at least one of the one-dimensional VASS-MDPs has linear angelic termination time, then $\A$ has also linear angelic termination time (using the same strategy).
\end{lemma}

In order to obtain the result for at least quadratic termination, we use the Azuma inequality for all the submartingales obtained from the one-dimensional VASS-MDPs.\\

Let $m_j^{(0)},m_j^{(1)}$ be the submartingale for $\A_j$, $\bar{Z}_j = \max_{q \in Q} \bar{z}_q$ obtained from the corresponding linear program (and assuming all values $\bar{z}_q$ are non-negative).\\

Given an initial configuration $p(n,\dots,n)$, the probability that the $j$-th counter decreases below zero in $t$ steps can be bounded from above:
\[
P(m_j^{(t)} - m_j^{(0)} \leq  -n+Z_j) \leq \mathrm{exp}(-(n-Z_j)^2 / t\alpha)
\]
where $Z_j$ and $\alpha$ are constants independent of $n$.

Let $Z = \max_{j=1,\dots,d} Z_j$ and $n \geq 2Z$, then:
\[
P(m_j^{(t)} - m_j^{(0)} \leq  -n+Z) \leq \mathrm{exp}(-(n-Z_j)^2 / t\alpha) \leq \mathrm{exp}(-n^2 / 4t\alpha).
\]

Observe that there exists a suitable constant $c > 0$ such that 
\[
d \cdot \mathrm{exp}(-c / 4\alpha) < 1
\]
since $d$ and $\alpha$ are constants (depending only on $\A$). Therefore, taking $t = n^2/c$, we obtain that the probability of some counter decreasing below zero is $1-\delta$ for some $\delta > 0$, and thus $\term_a(n) \in \Omega(n^2)$.

 Let $\eps >0$ and $t = n^{2-\eps}$, then
\[
\lim_{n\to \infty} P(m_j^{(t)} - m_j^{(0)} \leq  -n+Z_j) \leq \lim_{n \to \infty}\mathrm{exp}(-(n-Z_j)^2 / n^{2-\eps}\alpha) = 0.
\]
This completes the proof of Theorem~\ref{thm-main-angel}.

\subsection{A proof of Theorem~\ref{thm:general-main}}\label{app-MEC}
We can divide the set of states of $\A$ into the states belonging to some MEC and transient states. We rely on the two following facts:
\begin{enumerate}
\item For every strategy, the expected number of transitions from a transient state to some MEC state can be bounded by a constant $k$ (a number dependent only on $\A$ and not the size of the initial configuration).
\item The asymptotic complexity of a MEC does not depend on the initial state.
\end{enumerate}

First, we consider the demonic case. In DAG-like VASS-MDP, the only loops in the MEC decomposition are self-loops, i.e., once we leave MEC $M$ and visit a different MEC $M'$, we may never return to $M$. Moreover, there is a probability $p < 1$ such that for every MEC $M$ and every strategy, we revisit $M$ after leaving it (i.e., execute the self loop on $M$) with probability at most $p$.

For the ``if'' direction, we assume that the initial state is in a MEC $M$. We compute the upper bound on the expected number of transitions before terminating or arriving into another MEC. Let $Q'$ be the states of every MEC different from $M$.
We know that if MEC $M$ is linear then there exists $\bw_M > 0$ such that all increments $\bi_1,\dots,\bi_r$ in $M$ satisfy $\bi_j \cdot \bw_M < 0$.
 Let us consider a $\Q$-labeled MDP $\A_{\bw}$ obtained from $\A$ by replacing each label $\bu \in \Z^d$ by $\bu \cdot \bw_M \in \Q$.
 
Now we construct a supermartingale similar to the one in the proof of Lemma~\ref{lem-MDP-finite}. Again, for every $i \in \N$, let $S^{(i)}$ and $C^{(i)}$ be functions assigning to every infinite path $\pi = p_0,u_0,p_1,u_2,\dots$ in $\A_{\bw}$ initiated in $p$ the state $p_i$, and the sum $\bw \cdot  \bn + \sum_{j=1}^i u_j$ (where $\bn = (n,n,\dots,n)$ is the initial counter vector in $\A$). Furthermore, let $M.steps(i)$ and $T.steps(i)$ be functions counting for every infinite path $\pi$ the number of transitions in the MEC $M$ and in the transient states before entering a MEC different from $M$ (a transition $t = (q,\bu,q')$ is in $M$ if both $q,q'$ are in $M$).

Let $p$ be any initial state and $\sigma$ arbitrary strategy. Then the following sequence of random variables is a  supermartingale:
\[
m^{(i)} = \begin{cases}
C^{(i)} + \bar{z}_{S^{(i)}} - M.steps(i) \cdot \bar{x} + T.steps(i) \cdot K & \text{if } C^{(i)} \geq 0 \text{ and } S^{(i)} \neq Q'\\
m^{(i-1)} & \text{otherwise}
\end{cases}
\]
where $K$ is sufficiently large constant. We want to compute the expected value of $M.steps$. For every $i$ we have:
\[
\Exp_p^{\sigma}(m^{(i)}) = \Exp_p^{\sigma}(C^{(i)} + \bar{z}_{S^{(i)}} - M.steps(i) \cdot \bar{x} + T.steps(i) \cdot K).
\]
We know that $\Exp_p^{\sigma}(T.steps) \leq k$ and $\Exp_p^{\sigma}(\bar{z}_{S^{(i)}})$ is bounded by a constant. Moreover, $m^{(0)} \geq \bw \cdot \bn + K_1$ where $K_1$ is a constant depending only on $\A_{\bw}$. Using the property of a supermartingale, we obtain
\[
\Exp_p^{\sigma}(m^{(i)}) = \Exp_{p}^{\sigma}(C^{(i)}) + K_2 - \Exp_{p}^{\sigma}(M.steps(i)) \cdot \bar{x} \leq \bw\cdot \bn.
\]
Since $\bar{x} < 0$, we have for every $i \in \N$ that 
\[
\Exp_p^{\sigma}(M.steps(i)) \leq (\bw \cdot \bn -K_2 - \Exp_p^{\sigma}(C^{(i)}))/ |\bar{x}|,
\]
therefore $\Exp_p^{\sigma}(M.steps) + \Exp_p^{\sigma}(T.steps) \leq c\cdot n$ for a suitable constant $c$.\\

The time spent in one MEC can be used to increase some counters that can be used by MECs visited later. However, once we visit MEC $M'$, we never return to $M$. 
We define the height of MEC $M$ to be the length of a longest path in the MEC decomposition from $M$ into a bottom MEC (if MEC $M'$ can be visited from $M$, then $M'$ has higher height).

Let $\maxup = \max\{\|\bu\|; (q,\bu,q') \in T\}$ be the size of maximum counter change per transition.

We assume that for every MEC, the demonic termination complexity is bounded by $rn$ for all $n \in \N$. Let $i$ be the height of the MEC containing the initial state (or $i$ is such that $i-1$ is the height of a reachable MEC with the largest height). By induction on $i$, we prove that the expected termination time is bounded by $(\maxup \cdot r+1)^{i} \cdot n$ for $n$ sufficiently high.

If $i=0$, then $\term_d(n) \leq rn$.

Assume that the height of MEC $M$ containing the initial configuration is $i+1$ and and the expected termination time for $i$ is bounded by $(r+1)^{i} \cdot n$ (if we start in some transient state, the expected number of steps into a MEC with height at most $i$ is constant and the induction step holds). 
We divide every path $\pi = \pi_1\pi_2$ where $\pi_1$ is the part of $\pi$ prior the arrival into the lower MEC. 
The expected length of $\pi_1$ is bounded by $rn$. Therefore, we have the following upper bound on the expected size of counters when arriving into the lower MEC: $n + rn\cdot \maxup = (\maxup \cdot r +1) \cdot n$.

Using induction hypothesis, the expected length of $\pi$ is then $(\maxup \cdot r + 1)^{i+1} \cdot n$ which completes the proof.\\

For the ``only if'' direction, consider an initial state to be in a MEC with at least quadratic termination complexity. Using the corresponding strategy, we obtain the result.\\

Now we turn to the angelic case. If all bottom MECs are linear, there exists a strategy reaching one of the MECs in expected constant time. Therefore, the complexity is the same as in the bottom MEC, i.e., linear. If some of the bottom MECs is at least quadratic, then starting in that MEC, we obtain at least quadratic termination complexity.

\end{document}